\documentclass[leqno]{article}
\usepackage{amssymb,latexsym}
\usepackage{color,graphicx,itemlist}
\usepackage{epic,eepic}
\usepackage{hyperref}
\usepackage{showlabels}

\newcommand{\dwedge}[2]%
{\overset{#1}{\underset{#2}{\mbox{$\bigwedge\hspace{-2ex}\bigwedge$}}}}
\newcommand{\dvee}[2]%
{\overset{#1}{\underset{#2}{\mbox{$\bigvee\hspace{-2ex}\bigvee$}}}}

\newcommand{\inMath}[1]{\relax\ifmmode{#1}%
\else{\mbox{$#1$}}\fi}

\newcommand{\vdot}%
{\unitlength0.4mm
\begin{picture}(8,10)
\put(4,0){.}
\put(0,10){.}
\put(8,10){.}
\put(2,5){.}
\put(6,5){.}
\end{picture}
}

\newcommand{\ignore}[1]{}

\newcommand{\romanref}[1]{%
\if \ref{#1}\empty {\setcounter{romanrefcounter}0} \else 
{\setcounter{romanrefcounter}{\ref{#1}}}\fi%
{\it \roman{romanrefcounter}}}

\newcommand{\CA}{{\cal A}}

\newcommand{\BG}{{\bf G}}

\newcommand{\BL}{{{\bf L}}}

\newcommand{\BM}{{\bf M}}

\newcommand{\Bn}{{\bf n}}

\newcommand{\BS}{{\bf S}}

\newcommand{\tO}{\twoheadrightarrow}

\newlength{\circlen}
\newlength{\symblen}

\ifx\text\undefined
  \newcommand{\text}[1]{\relax
    \ifmmode\mathchoice
      {\hbox{\the\textfont0\relax#1}}%
      {\hbox{\the\textfont0\relax#1}}%
      {\hbox{\the\scriptfont0\relax#1}}%
      {\hbox{\the\scriptscriptfont0\relax#1}}%
    \else{\relax#1}\fi}
  \fi

\newcommand{\defsub}[1]{}

\def\twoheaddownarrow{\rlap{$\downarrow$}\raise-.5ex\hbox{$\downarrow$}}
\def\twoheaduparrow{\rlap{$\uparrow$}\raise.5ex\hbox{$\uparrow$}}

\def\texturespicture #1 by #2 (#3){
\vbox to #2 {\hrule width #1 height 0pt depth 0pt
}}

\def\scaledpicture #1 by #2 (#3 scaled #4){{\dimen0=#1 \dimen1=#2
\divide\dimen0 by 1000 \multiply \dimen0 by #4
\divide\dimen1 by 1000 \multiply \dimen1 by #4
\texturespicture\dimen0 by \dimen1 (#3 scaled #4)}}

 {\begin{tabular}[t]{llr} && \fbox{#1}\\}%
 {\end{tabular}}

{\end{list}}

{\end{tabular}%
\end{trivlist}}

\def\introrule{{\cal I}}\def\elimrule{{\cal E}}

\newdimen\proofrulebreadth \proofrulebreadth=.05em
\newdimen\proofdotseparation \proofdotseparation=1.25ex
\newdimen\proofrulebaseline \proofrulebaseline=2ex
\newcount\proofdotnumber \proofdotnumber=3
\let\then\relax
\def\hfi{\hskip0pt plus.0001fil}
\mathchardef\squigto="3A3B
\newif\ifinsideprooftree\insideprooftreefalse
\newif\ifonleftofproofrule\onleftofproofrulefalse
\newif\ifproofdots\proofdotsfalse
\newif\ifdoubleproof\doubleprooffalse
\let\wereinproofbit\relax
\newdimen\shortenproofleft
\newdimen\shortenproofright
\newdimen\proofbelowshift
\newbox\proofabove
\newbox\proofbelow
\newbox\proofrulename
\def\shiftproofbelow{\let\next\relax\afterassignment\setshiftproofbelow\dimen0 }
\def\shiftproofbelowneg{\def\next{\multiply\dimen0 by-1 }%
\afterassignment\setshiftproofbelow\dimen0 }
\def\setshiftproofbelow{\next\proofbelowshift=\dimen0 }
\def\setproofrulebreadth{\proofrulebreadth}

\def\prooftree{
\ifnum  \lastpenalty=1
\then   \unpenalty
\else   \onleftofproofrulefalse
\fi
\ifonleftofproofrule
\else   \ifinsideprooftree
        \then   \hskip.5em plus1fil
        \fi
\fi
\bgroup
\setbox\proofbelow=\hbox{}\setbox\proofrulename=\hbox{}%
\let\justifies\proofover\let\leadsto\proofoverdots\let\Justifies\proofoverdbl
\let\using\proofusing\let\[\prooftree
\ifinsideprooftree\let\]\endprooftree\fi
\proofdotsfalse\doubleprooffalse
\let\thickness\setproofrulebreadth
\let\shiftright\shiftproofbelow \let\shift\shiftproofbelow
\let\shiftleft\shiftproofbelowneg
\let\ifwasinsideprooftree\ifinsideprooftree
\insideprooftreetrue
\setbox\proofabove=\hbox\bgroup$\displaystyle 
\let\wereinproofbit\prooftree
\shortenproofleft=0pt \shortenproofright=0pt \proofbelowshift=0pt
\onleftofproofruletrue\penalty1
}

\def\eproofbit{
\ifx    \wereinproofbit\prooftree
\then   \ifcase \lastpenalty
        \then   \shortenproofright=0pt  
        \or     \unpenalty\hfil         
        \or     \unpenalty\unskip       
        \else   \shortenproofright=0pt  
        \fi
\fi
\global\dimen0=\shortenproofleft
\global\dimen1=\shortenproofright
\global\dimen2=\proofrulebreadth
\global\dimen3=\proofbelowshift
\global\dimen4=\proofdotseparation
\global\count255=\proofdotnumber
$\egroup  
\shortenproofleft=\dimen0
\shortenproofright=\dimen1
\proofrulebreadth=\dimen2
\proofbelowshift=\dimen3
\proofdotseparation=\dimen4
\proofdotnumber=\count255
}

\def\proofover{
\eproofbit 
\setbox\proofbelow=\hbox\bgroup 
\let\wereinproofbit\proofover
$\displaystyle
}%
\def\proofoverdbl{
\eproofbit 
\doubleprooftrue
\setbox\proofbelow=\hbox\bgroup 
\let\wereinproofbit\proofoverdbl
$\displaystyle
}%
\def\proofoverdots{
\eproofbit 
\proofdotstrue
\setbox\proofbelow=\hbox\bgroup 
\let\wereinproofbit\proofoverdots
$\displaystyle
}%
\def\proofusing{
\eproofbit 
\setbox\proofrulename=\hbox\bgroup 
\let\wereinproofbit\proofusing
\kern0.3em$
}

\def\endprooftree{
\eproofbit 
  \dimen5 =0pt
\dimen0=\wd\proofabove \advance\dimen0-\shortenproofleft
\advance\dimen0-\shortenproofright
\dimen1=.5\dimen0 \advance\dimen1-.5\wd\proofbelow
\dimen4=\dimen1
\advance\dimen1\proofbelowshift \advance\dimen4-\proofbelowshift
\ifdim  \dimen1<0pt
\then   \advance\shortenproofleft\dimen1
        \advance\dimen0-\dimen1
        \dimen1=0pt
        \ifdim  \shortenproofleft<0pt
        \then   \setbox\proofabove=\hbox{%
                        \kern-\shortenproofleft\unhbox\proofabove}%
                \shortenproofleft=0pt
        \fi
\fi
\ifdim  \dimen4<0pt
\then   \advance\shortenproofright\dimen4
        \advance\dimen0-\dimen4
        \dimen4=0pt
\fi
\ifdim  \shortenproofright<\wd\proofrulename
\then   \shortenproofright=\wd\proofrulename
\fi
\dimen2=\shortenproofleft \advance\dimen2 by\dimen1
\dimen3=\shortenproofright\advance\dimen3 by\dimen4
\ifproofdots
\then
        \dimen6=\shortenproofleft \advance\dimen6 .5\dimen0
        \setbox1=\vbox to\proofdotseparation{\vss\hbox{$\cdot$}\vss}%
        \setbox0=\hbox{%
                \advance\dimen6-.5\wd1
                \kern\dimen6
                $\vcenter to\proofdotnumber\proofdotseparation
                        {\leaders\box1\vfill}$%
                \unhbox\proofrulename}%
\else   \dimen6=\fontdimen22\the\textfont2 
        \dimen7=\dimen6
        \advance\dimen6by.5\proofrulebreadth
        \advance\dimen7by-.5\proofrulebreadth
        \setbox0=\hbox{%
                \kern\shortenproofleft
                \ifdoubleproof
                \then   \hbox to\dimen0{%
                        $\mathsurround0pt\mathord=\mkern-6mu%
                        \cleaders\hbox{$\mkern-2mu=\mkern-2mu$}\hfill
                        \mkern-6mu\mathord=$}%
                \else   \vrule height\dimen6 depth-\dimen7 width\dimen0
                \fi
                \unhbox\proofrulename}%
        \ht0=\dimen6 \dp0=-\dimen7
\fi
\let\doll\relax
\ifwasinsideprooftree
\then   \let\VBOX\vbox
\else   \ifmmode\else$\let\doll=$\fi
        \let\VBOX\vcenter
\fi
\VBOX   {\baselineskip\proofrulebaseline \lineskip.2ex
        \expandafter\lineskiplimit\ifproofdots0ex\else-0.6ex\fi
        \hbox   spread\dimen5   {\hfi\unhbox\proofabove\hfi}%
        \hbox{\box0}%
        \hbox   {\kern\dimen2 \box\proofbelow}}\doll%
\global\dimen2=\dimen2
\global\dimen3=\dimen3
\egroup 
\ifonleftofproofrule
\then   \shortenproofleft=\dimen2
\fi
\shortenproofright=\dimen3
\onleftofproofrulefalse
\ifinsideprooftree
\then   \hskip.5em plus 1fil \penalty2
\fi
}

\newcommand\strikethrough[1]{{\setbox0=\hbox{$#1$}
\hrule height.75ex depth-.65ex width\wd0 \kern-\wd0\box0}}

\mathchardef\gt="313E \mathchardef\lt="313C

\def\undern#1{\vtop{\ialign{##\crcr
 $\hfil\displaystyle{#1}\hfil$\crcr
 \noalign{\kern-.1pt\nointerlineskip}%
 ~\,\raisebox{-.5ex}{$n \to \infty$} \crcr}}}

\def\undern#1{\vtop{\ialign{##\crcr
 $\hfil\displaystyle{#1}\hfil$\crcr
 \noalign{\kern-.1pt\nointerlineskip}%
 ~\,\raisebox{-.5ex}{$i$} \crcr}}}

\usepackage{proof,amsmath,enumitem}
\usepackage{tikz}
\usetikzlibrary{arrows}

\reversemarginpar

 \newtheorem{theorem}{Theorem}[section]					
 \newtheorem{definition}[theorem]{Definition}				
 	
\newtheorem{remark}[theorem]{Remark}					
\newtheorem{example}[theorem]{Example}					

\newenvironment{proof}{\begin{trivlist}\item[]{\bf Proof.}}{\hspace*{\fill} $\blacksquare$ \end{trivlist}}

\newcommand\upa[1]{\{ #1\}}

\newcommand\DA{\Delta_{\mathcal A}}
\newcommand\DAp{\Delta_{\mathcal A'}}
\newcommand\In{\text{\it In}}

\newcommand\pG{\pi{\bf G}_3}
\newcommand\bn{{\bf n}}
\newcommand\bm{{\bf m}}
\newcommand\bs{{\bf s}}
\newcommand\lbm{\lambda_\bm}

\newcommand\OA{\Omega_{\CA}}
\newcommand\PA{\Psi_{\CA}}

\newcommand\rhod{\mathbin{\rho_{\text{\tiny $\lor$}}}}

\newcommand\rhoc{\mathbin{\rho_{\text{\tiny $\land$}}}}

\newcommand\BLR{\BL(R)}

\begin{document}
\title{The Attack as Intuitionistic Negation}

\date{Compiled on \today} 

\author{D. Gabbay\\
King's College London,\\ Department of Informatics,\\
Strand, London, WC2R 2LS, UK;\\
Bar Ilan University, Ramat Gan, Israel\\
and  \\
University of Luxembourg, Luxembourg.\\
{\tt dov.gabbay@kcl.ac.uk}\\[2.5ex]
M. Gabbay\\
Cambridge University, UK.\\
{\tt mg639@cam.ac.uk}\\
}
\maketitle

\begin{abstract}
We translate the argumentation networks $\CA=(S, R)$ into a theory $\DA$ of intuitionistic logic, retaining $S$ as the domain and using intuitionistic negation to model the attack $R$ in $\CA$: the attack $xRy$ is translated to $x\to\neg y$. The intuitionistic models of $\DA$ characterise the complete extensions of $\CA$.

The reduction of argumentation networks to intuitionistic logic yields, in addition to a representation theorem, some additional benefits: it allows us to give semantics to higher level attacks, where an attack ``$xRy$'' can itself attack another attack ``$uRv$''; one can make higher level meta-statements $W$ on $(S, R)$ and such meta-statements can attack and be attacked in the domain.
\end{abstract}

\section{Background and orientation}\label{sec1}
This paper is a continuation of \cite{541-1} but it is self-contained and is independent of \cite{541-1}, except that it expands the ideas of \cite{541-1}.

Given a finite abstract argumentation network $\CA=(S, R)$, where $S\neq \varnothing$ is the set of arguments and $R\subseteq S\times S$ is the attack relation, we would like to view the set $S$ as atomic propositions of the intuitionistic propositional calculus and translate the attack relation $xRy$ as $x\to \neg y$, where ``$\to$'' represents intuitionistic implication and ``$\neg$'' represents intuitionistic negation. For each $\CA$ we write a theory $\Delta_{\CA}$ such that all the complete extensions of $\CA$ correspond to all the intuitionistic models of $\Delta_{\CA}$.

The reduction of argumentation networks to intuitionistic logic yields, in addition to a representation theorem, some additional benefits.
\begin{itemize}
\item
It allows us to give semantics to higher level attacks, where an attack ``$xRy$'' can itself attack another attack ``$uRv$''.
\item
One can make higher level meta-statements $W$ on $(S, R)$ and such meta-statements can attack and be attacked.
\end{itemize}
%
For example we can attack an argument $a$ by saying that $a$ is a generic argument which attacks all other arguments $\{x\mid x\neq a\}$ and therefore $a$ should be out.

What we are saying is $W(a)$ where:
\[
W(a) = \forall x(x\neq a\rightarrow aRx)
\]
attacks $a$.

We shall use G\"{o}del's intuitionistic logic $\BG_3$, semantically defined by all intuitionistic Kripke models with just two linearly ordered worlds $t < s$ ($t$ the actual world and $s$ a possible world), with $<$ the intuitionistic accessibility relation. Appendix A describes the logic $\BG_3$ in detail.  It can be axiomatised.

We present the complete extensions of $\CA = (S, R)$, using the Caminada labelling approach, \cite{541-2}.  A Caminada labelling of $S$ is a function $\lambda: S \mapsto \{\mbox{in, out, und}\}$ such that the following holds

\begin{itemlist}{CCCC}
\item [(C1)] $\lambda (x)=$ in iff for all $y$ attacking $x$, $\lambda (y)=$ out.
\item [(C2)] $\lambda(x)=$ out iff for some $y$ attacking $x$, $\lambda (y)=$ in.
\item [(C3)] $\lambda (x)=$ und iff for all $y$ attacking $x$, $\lambda (y)\neq$ in, and for some $y_0$ attacking $x$, $\lambda (y_0)=$ und.\footnote{That is, $\lambda (x)=$ und iff neither $\lambda (x)=$ in nor $\lambda (x)=$ out.}
\item [(C4)] If $x$ is not attacked at all, then $\lambda (x)=$ in.\footnote{This condition follows from (C1), since the empty conjunction is considered $\top$.}
\end{itemlist}

Let us use the following notation for $\BG_3$.

For a proposition $e$, write $e=(\top, \top)$ to mean $t\vDash e$ and $s\vDash e$.  Write $e=(\bot,\top)$ to mean $t\not\vDash e$ and $s\vDash e$.  Write $e=(\bot,\bot)$ to mean $t\not\vDash e$ and $s\not\vDash e$.

Note that since $\{t < s\}$ is an intuitionistic model, the option $t\vDash e$ and $s\not\vDash e$ is not allowed.

We denote assignments $h$ of truth values to atoms in the model by $h(e)=(\top,\top)$ or $h(e)=(\bot,\top)$ or $h(e)=(\bot,\bot)$.  We also write $e=(\top,\top),(\bot,\top),(\bot,\bot)$ respectively, using abuse of notation, when the assignment is known.  So $t\vDash e$ means $t\vDash_h e$ and $s\vDash e$ means $s\vDash _h e$.

\section{Translation into $\BG_3$}\label{sec2}
Let $\CA=(S, R)$ be a finite argumentation network. Let $\BG_3$ be G\"{o}del's logic with two possible worlds $t < s$, with $t$ the actual world.  Let $S$ be considered as propositional atoms of $\BG_3$, and let \Bn\ be an additional new constant of $\BG_3$, whose truth condition is $(\bot, \top)$ meaning $t\not\vDash \Bn$ and $s\vDash \Bn$. We thus have that $t\not\vDash \Bn\vee\neg\Bn$ and $t\not\vDash \neg\neg\Bn\to \Bn$.

We shall use $\Bn$ in our translation and eliminate it in the next section.  \Bn\ can be characterised by the condition $\not\vDash \Bn\vee\neg\Bn$.  In other words , we take an axiom system for $\BG_3$ and add to the language the constant \Bn\ with the ``axiom" that  $\Bn \vee \neg  \Bn$  is false.

\begin{definition}\label{541-D1}
Let $\CA =(S, R)$ and \Bn\ be as described before.  Define the theory $\DA$ containing  the following formulas of $\BG_3$, for $x,y$ in $S$:

\begin{itemlist}{CCCC}
\item [(a1)] $\bigwedge_{x\in S} (x\to \Bn\vee\bigwedge_{yRx}\neg y)$
\item [(a2)] $\bigwedge _{x\in S}(\bigwedge_{yRx} \neg y \to \Bn\vee x)$
\item [(b1)] $\bigwedge_{x\in S}(\neg x\to \Bn\vee\bigvee_{yRx}y)$
\item [(b2)] $\bigwedge_{x\in S}(\bigvee_{yRx} y\to \neg x\vee\Bn)$
\end{itemlist}
\end{definition}

\begin{theorem}\label{541-T2}{\ }
\begin{enumerate}
\item Let $h$ be a model of $\Delta_{\CA}$.  Define a Caminda labelling $\lambda_h$ as follows:
\begin{itemize}
\item $\lambda_h(x)=$ in, if $h(x)=(\top,\top)$
\item $\lambda_h(x)=$ out, if $h(x)=(\bot,\bot)$
\item $\lambda_h(x)=$ und, if $h(x) =(\bot,\top)$.
\end{itemize}
Then $\lambda_h$ is a complete extension for $\CA$.
\item Let $\lambda$ be a complete extension of $\CA$.  Define $h_\lambda(x)$ for $x\in S$ by
\begin{itemize}
\item  $h_\lambda (x) =(\top,\top)$, if $\lambda (x) =$ in.
\item $h_\lambda (x) =(\bot,\bot)$, if $\lambda (x)=$ out
\item $h_\lambda (x) =(\bot,\top)$, if $\lambda (x)=$ und.
\end{itemize}
Then $h_\lambda$ is a model of $\Delta_{\CA}$.
\end{enumerate}
\end{theorem}

\begin{proof}{\ }
\paragraph{Part 1.}  Assume $h$ is a model of $\Delta_{\CA}$.  We show $\lambda_h$ satisfies the conditions of Caminada labelling.
\paragraph{Condition C1.}  Assume $\lambda_h(x)=$ in. Then $x=(\top,\top)$.  From axiom (a1) we get, since $t\vDash x$ and $t\not\vDash \Bn$ that $t\vDash \bigwedge_{yRx}\neg y$ and so $h(y)=(\bot,\bot)$ for all attackers $y$ of $x$, i.e.\ $\lambda_h(y)=$ out for all attackers of $x$.

Assume $\lambda_h(x)=$ out for all attackers of $x$. Then $h(\bigwedge_{yRx}\neg y)=(\top,\top)$. From axiom (a2) we get $t\vDash x$ and so $h(x) =(\top,\top)$ and so $\lambda_h(x)=$ in.

\paragraph{Condition C2.}  Assume $\lambda_h(x)=$ out. Then $x=(\bot,\bot)$. Then from axiom (b1) we have $t\vDash \bigvee_{yRx} y$. Let $y_0$ be such that  $t\vDash y_0$.  Then $y_0=(\top,\top)$ and so $\lambda_h(y)=$ in.

Assume for some $y_0$ we have $\lambda_h(y_0)=$ in. Then $y_0=(\top,\top)$ and so from axiom (b2) we get that $t\vDash \neg x\vee\Bn$. Therefore $t\vDash \neg x$ and so $x=(\bot,\bot)$ and so $\lambda_h(x)=$ out.

\paragraph{Condition C3.}  We have $\lambda_h(x)=$ und iff by definition $h(x)=(\bot,\top)$.
From (b2) we have that
\[
t\vDash (\bigvee_{yRx} y\to \neg x \vee \Bn).
\]
But at $t$ we have that $t\not\vDash \neg x\vee\Bn$. So $t\not\vDash \bigvee_{yRx} y$. So all attackers $y$ are false at $t$. So their value can be either  $(\bot,\bot)$ (i.e.\ out) or $(\bot,\top)$ (i.e.\ undecided) but not $(\top,\top)$.

Assume $x=(\bot,\top)$.  From (a2) we have that
\[
t\vDash (\bigwedge_{xRy} \neg y \to \Bn\vee x).
\]
Since $t\not\vDash \Bn\vee x$ we have $t\not\vDash \bigwedge_{yRx}\neg y$.  Thus there exists at least one $y_0$ such that $t\not\vDash \neg y_0$.  So either $y_0= (\top,\top)$ or $(\bot,\top)$. By what we have shown $y_0$ must be $(\bot,\top)$, i.e.\ $y_0$ is undecided (und).

Assume now that all attackers of $x$ are either out or und and one attacker $y_0$ is und. Thus we have all $y$ attacking $x$ are either $(\bot,\bot)$ or $(\bot,\top)$ and one $y_0=(\bot,\top)$.

From (b1) we have
\[
t\vDash \neg x\to \Bn\vee\bigvee_{yRx} y.
\]
Since $t\not\vDash \Bn \bigvee_{yRx} y$ we must have $t\not\vDash \neg x$. So $x=(\top,\top)$ or $(\bot,\top)$.

Finally from (a1) we have
\[
t\vDash x \to \Bn\vee\bigwedge_{yRx} \neg y.
\]
Since $y_0=(\bot,\top)$ we get $t\not\vDash \Bn\vee\bigwedge_{yRx} \neg y$ and so $t\not\vDash x$.
Therefore $x=(\bot,\top)$, i.e.\ $x$ is und.

\paragraph{Part 2.}  Assume $\lambda$ is a Caminada labelling satisfying (C1), (C2) and (C3). We show that $h_\lambda$ is a model of $\Delta_{\CA}$.  
 Since $s\vDash \Bn$, the axioms clearly hold at $s$. Therefore need only show the the axioms of $\Delta_{\CA}$ holds at $t$, this is made simpler by the fact that $t\nvDash\Bn$.

\paragraph{Axiom (a1).}  Show if $t\vDash x$ then $\vDash \bigwedge_{yRx} \neg y$.  This follows from condition (C1).  $t\vDash x$ means $\lambda (x)=$ in. So $\lambda (y)=$ out for all $yRx$ so $t\vDash \bigwedge_{yRx} \neg y$.

\paragraph{Axiom (a2).}  Show that if $t\vDash \bigwedge_{yRx} \neg y$ then $t\vDash x$. This again follows from (C1).

\paragraph{Axiom (b1).}  Show that if $t\vDash \neg x$ then $t\vDash \bigvee_{yRx} y$. 

This follows from (C2).  If $t\vDash \bigvee_{yRx} y$ then $t\vDash y_0$ for some $y_0$.  So $h_\lambda (y_0=(\top,\top)$. So $\lambda (y_0)=$ in. So $\lambda (x)=$ out. So $h_\lambda (x) =(\bot,\bot)$ so $t\vDash \neg x$

\paragraph{Axiom (b2).}  Show that if $t\vDash \bigvee_{yRx} y$ then $t\vDash \neg x$. This again follows from (C2).
\end{proof}

\section{Refinements and remarks}

\begin{remark}\label{541-R3}
We used the constant \Bn. We now try to eliminate it.  Let $\CA=(S, R)$ be given.  Define \Bn\ to be
\[
\Bn =\bigwedge_{x\in S} (x\vee\neg x).
\]
Assume for the moment that we are dealing with extensions which are {\em not} stable. So there exist undecided elements in the extension $\lambda$, i.e.\ there exists some $x$ such that $\lambda(x)=$ und and so $h(x) =(\bot,\top)$. So $x\vee\neg x$ is false at $t$ and true in $s$. So \Bn\ is false at $t$ and true in $s$.  So we found the \Bn\ we need for our translation.

We now offer a translation $\Theta^\lambda_{\CA}$ of an intuitionistic theory corresponding to any labelling $\lambda$. The translation is done by case analysis on $\lambda$.

\begin{enumerate}
\item Case $\lambda$ does not give value und., i.e.\ $\lambda(x) $ is either in or out.  This means that $\lambda$ represents a stable extension.
Let $\Theta^0_{\CA} = \bigwedge_x (x\leftrightarrow \bigwedge_{yRx} \neg y)$.
\item Case $\lambda$ gives value undecided. Let $\Theta^1_{\CA}=\Delta_{\CA}$ where the  \Bn\  which is used in $\Delta_{\CA}$ is $\Bn=\bigwedge_{x\in S} (x\vee\neg x)$.
\end{enumerate}

 Define $h_\lambda$ as in Section 2. We have that $\lambda$ is a complete extension iff $\Theta_{\CA}$ holds under $h_\lambda$.

The translation and the two case analyses described above  also work in the other direction.   Let $h$ be a model. Define $\lambda_h$ as in Section 2. We distinguish two cases.
\paragraph{Case $t\vDash_h\Bn$.}  In this case we look at $\Theta^0_{\CA}$ and indeed, $t\vDash \Theta^0_{\CA}$ off $\lambda_h$ is a stable complete extension.
\paragraph{Case $t\not\vDash_h\Bn$.}  In this case look at $\Theta^1_{\CA} =\Delta_{\CA}$.  Indeed $\lambda_h$ is a non-stabe complete extension iff $t\vDash_h\Delta_{\CA}$.
\end{remark}

\begin{remark}\label{541-R4}
The perceptive reader might argue as follows:
\begin{quote}
``What is the big deal about interpreting abstract argumentation in intuitionistic logic $\BG_3$?  There are many such translations into many diverse systems.  Your own paper \cite{541-1} gave another interpretation in classical logic.
Grossi has one in modal logic, there are many others as indeed quoted in \cite{541-1}. What is really going on here is that the Caminada labelling requires 3 values, in, out and undecided. Any logic which can isolate 3 values for ``in'', ``out'' and ``undecided'' can do the job!?  So what is the big deal here?"
\end{quote}  The answer to the question of what is the  ``big deal" or added value or difference between translations  lies in allowing for different possibilities for generalisations, extensions, variations, instantiations, or any other operations we do on argumentation networks.  Different translations would react differently in different environments and interpretations.

Let us illustrate by an example. Take the two very similar translations proposed in our own papers, one is the translation  into classical logic of \cite{541-1} and the other is the present translation into intuitionistic logic .

We will show a substantially different behaviour in the context of instantiation. Instantiation is important, as any supporter of ASPIC (\cite{541-15,541-16,541-17,541-18}) will tell you. Let us take a very simple network with  $S =\{a,b,x\}$ and $R =\{ (a,b) , (b,a), (a,x)\}$. Let us instantiate $x$ with  $x= p \vee \neg p$. Consider Figure \ref{541-F5}

\begin{figure}
\centering
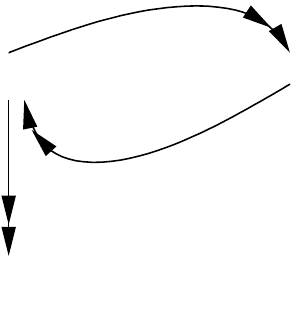
\caption{}\label{541-F5}
\end{figure}

In this figure, $x$ is instantiated by the wff $p\vee\neg p$.   When we translate into classical logic, $p\vee\neg p$ is $\top$ and so there can be only one extension.

\[
x =\top =\mbox{ in}, b=\mbox{ out}, a=\mbox{ in}.
\]

In intutionistic logic, $\BG_3, p\vee\neg p$ can be $(\top,\top)=$ in or $(\bot,\top)=$ und.  So the possible extensions are
\begin{center}
$a=$ in, $b=$ out, $x=\top=$in, \\
$b=$ in, $a=$ out, $x=$ out, \\
$a=b=x=$ und.
\end{center}
For further benefits see the next Remark~\ref{541-R6} and Sections~\ref{sec4} and~\ref{sec5} below.
\end{remark}

\begin{remark}\label{541-R6}
We assumed that $(S, R)$ is finite. this is not necessary. We can assume that $(S, R)$ is finitary, i.e.\ each $x$ has a finite number of attackers. This will allow us to write $\Delta_{\CA}$ as an infinite theory containing wffs for each $x\in S$.

The correspondence between extensions $\lambda$ and models $h$ will still work.

The challenge is when we have nodes with an infinite number of attackers. We cannot write the expression like this: $\bigwedge_{yRx}\neg y$.  It is not finite.  For this we need a predicate theory talking about domains $S$ with a binary relation $xRy$ on $S$ and a unary predicate $\In(x)$ on $S$ and write
\[
\forall y (yRx\to \neg\In(y)).
\]
The beauty of this language is that we can now also write
\[
\In(z) \to \neg xRy
\]
which means that $z$ attacks the attack from $x$ to $y$. This is a higher level attack. Furthermore, since $xRy$ is a formula of the logic, it can get value  $(\top,\top)$, i.e. it is in for both worlds; or $(\bot,\top)$, i.e. it is undecided (out in the first but in in the second); or $(\bot, \bot)$, i.e. it is out for both worlds.
The next section will do this systematically.

Furthermore, since $xRy$ is a formula of the language, it can also go on and attack, so we can write an attack attacking another attack:
\[
xRy \to\neg uRv
\]
and an attack attacking another node:
\[
xRy \to\neg \In(z)
\]
\end{remark}

The reader might ask, if we have the relation $R$ in the language, should we not have the following equivalence?
\[
xRy \leftrightarrow (\In(x) \to \neg \In(y))
\]
The answer is negative. $xRy$ represents a geometrical relation in  the network. It has a role in defining the attack conditions and the defining clauses for  $\{$in, out, und$\}$ values. The formula $\In(x) \to \neg \In(y)$ also participates in these defining clauses, but describes the properties of the geometrical relation as a whole. So for example we do want that
\[
xRy \to (\In(x) \to \neg \In(y))
\]
but the converse does not necessarily hold. For example suppose $(S,R)$ is as in Figure~\ref{541-mike1}. Then $\In(a) \to \neg \In(d)$ but it is not the case that $aRd$.
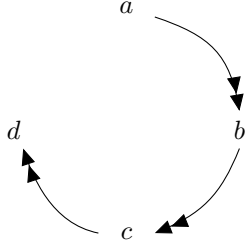
\begin{figure}[ht]
\centering
\ifx\JPicScale\undefined\def\JPicScale{1}\fi
\unitlength \JPicScale mm
\begin{tikzpicture}[x=\unitlength,y=\unitlength,inner sep=0pt]
\draw (15,30) node {$a$};
\draw (30,13.75) node {$b$};
\draw (0,13.75) node {$d$};
\draw [->>,>=triangle 45](18.75,28.75) .. controls (26.25,26.25) and (28.75,23.75) .. (30,16.25);
\draw (15,0) node {$c$};
\draw [<<-,>=triangle 45](1.25,11.25) .. controls (2.5,7.5) and (5,1.25) .. (11.25,0);
\draw [<<-,>=triangle 45](18.75,0) .. controls (25,2.5) and (27.5,5) .. (30,11.25);
\end{tikzpicture}
\caption{If $a$ is in then $b$ is out so $c$ is in and $d$ is out, but $a$ does not attack $d$.}\label{541-mike1}
\end{figure}

\section{The attack in Predicate Intuitionistic Logic.}\label{sec4}

This section informally develops the ideas put forward in Remark~\ref{541-R6}. The formal development will be in the next Section~\ref{sec5}. We adopt the language of intuitionistic predicate logic with the usual connectives and quantifiers $\{\land,\lor,\to,\neg,\bot,\top,\forall,\exists\}$ and the axioms for constant domains, namely:
\[
\forall x(A\lor B(x))\to A\lor\forall x B(x)
\]
where $x$ is not free in $A$.

Our logic, which we call $\pG$, is defined semantically using a two world predicate Kripke model of constant domains. The worlds are as before $\{t,s\}$, with $t< s$ and $t$ the actual world. We denote the domain of these worlds by $S$. This will also be our set of arguments. We allow two non-logical predicates in this language, a binary predicate $xRy$ and a unary predicate $\In(z)$. $R$ represents the attack relation and $\In$ represents being ``in'' the extension defined by the model.

Being an intuitionistic language, any formula $A$ has three options for values in the two world model:
\begin{itemize}
\item
$t\nvDash A$ and $s\nvDash A$
\item
$t\nvDash A$ and $s\vDash A$
\item
$t\vDash A$ and $s\vDash A$
\end{itemize}
We can write the values as $A=(\bot,\bot)$, $A=(\bot,\top)$ and $A=(\top,\top)$.

For the intuitionistic logic and semantics see the Appendix.
\begin{remark}\label{541-4R1}
This remark explains how the language with the predicates $\In(z)$ and $xRy$ can deal with infinite networks.

Let $(S,R)$ be an arbitrary network. So nodes $x\in S$ may have an infinite number of attackers. We want to use predicate logic to define the analogous formulas to those of Definition~\ref{541-D1}.

We let $\bn$ be a constant as before whose truth value is $(\bot,\top)$. We now define the analogous wffs $(A1)$, $(A2)$, $(B1)$, $(B2)$. We take a model with domain $S$. We assume of this model that it has a constant domain and we assume that $R$ is decided, namely:
\[
\forall xy(xRy\lor\neg xRy)
\]
The attack relation $R$ is expressible in the model because $R$ is in the language. We write:
\[
\forall x\bigl[\In(x)\to\bigl(\bn\lor\forall y(yRx\to\neg \In(y))\bigr)
\bigr] \tag{A1}
\]
\[
\forall x\bigl[ \forall y\bigl(yRx\to\neg\In(y)\bigr)\to(\bn\lor\In(x))
\bigr] \tag{A2}
\]
\[
\forall x\bigl[\neg \In(x)\to\bigl(\bn\lor\exists y(yRx\land \In(y))\bigr)
\bigr] \tag{B1}
\]
\[
\forall x\bigl[ \exists y(yRx\land\In(y))\to(\bn\lor\neg \In(x))
\bigr] \tag{B2}
\]
Then we get the following theorem~\ref{541-4T2}.
\end{remark}

\begin{theorem}\label{541-4T2}
\begin{enumerate}
\item
Let $\bm$ be a model of $\DA$. Define $\lbm$ as follows:
\begin{itemize}
\item
$\lbm(x)=\text{in}$, if $\bm(\In(x))=(\top,\top)$
\item
$\lbm(x)=\text{out}$, if $\bm(\In(x))=(\bot,\bot)$
\item
$\lbm(x)=\text{und}$, if $\bm(\In(x))=(\bot,\top)$.
\end{itemize}
Then $\lbm$ is a complete extension of $(S,R)$.
\item
Let $\lambda$ be a complete extension of $(S,R)$. Define a model $\bm$ with domain $S$ and a decided relation $R$. Define the values for $\In(x)$, $x\in S$ as follows:
\begin{itemize}
\item
$\In(x)=(\top,\top)$, if $\lambda(x)=\text{in}$
\item
$\In(x)=(\bot,\bot)$, if $\lambda(x)=\text{out}$
\item
$\In(x)=(\bot,\top)$, if $\lambda(x)=\text{und}$
\end{itemize}
Then $\bm$ is a model of $\DA$.
\end{enumerate}
\end{theorem}
\begin{proof}
Parallel to the proof of Theorem~\ref{541-T2}.
\end{proof}

\begin{remark}\label{541-4R3}
The inclusion of $R$ in the predicate language gives us the opportunity to express higher level attacks not only on nodes and other attack arrows but also on meta-level statements. Consider the higher level network of Figure~\ref{541-4F4}. In this figure we have:
\[
\begin{array}{cl}
(i) & a\tO b \\
(ii) & a\tO c \\
(iii) & a\tO (c\tO d) \\
(iv) & (a\tO b)\tO d \\
(v) & c\tO d
\end{array}
\]
\begin{figure}
\centering
\ifx\JPicScale\undefined\def\JPicScale{1}\fi
\unitlength \JPicScale mm
\begin{tikzpicture}[x=\unitlength,y=\unitlength,inner sep=0pt]
\draw (0,30) node {$a$};
\draw (30,30) node {$b$};
\draw (0,0) node {$c$};
\draw (30,0) node {$d$};
\draw [->>,>=triangle 45](2.5,30) -- (27.5,30);
\draw [->>,>=triangle 45](0,27.5) -- (0,2.5);
\draw [->>,>=triangle 45](2.5,0) -- (27.5,0);
\draw [->>,>=triangle 45](1.25,28.75) -- (15,1.25);
\draw [->>,>=triangle 45](15,30) -- (28.75,2.5);
\end{tikzpicture}
\caption{}\label{541-4F4}
\end{figure}
$(i)$ and $(ii)$ and $(v)$ are ordinary attacks. $(iii)$ is a higher level attack first introduced in~\cite{541-14} and further studied extensively in~\cite{541-10},~\cite{541-11},~\cite{541-12} and~\cite{541-13}. $(iv)$ was introduced in~\cite{541-14} and has never had systematic semantics until the current paper.

The above attacks can be translated using $\In(z)$ and $xRy$ as follows: we need to give up the axiom $\forall xy(xRy\lor\neg xRy)$ and allow $xRy$ to get the undecided value $(\bot,\top)$. Because of the possibility that an attack arrow $xRy$ may be out or undecided we need to express explicitly in the clauses below that the relation holds, we can now write:
\[
\begin{array}{cl}
(i') & (\In(a)\land aRb)\to \neg\In(b) \\
(ii') & (\In(a)\land aRc)\to \neg\In(c)\\
(iii') & \In(a)\to \neg cRd \\
(iv') & aRb\to \neg\In(d) \\
(v') &  (\In(c) \land cRd) \to \neg\In(d)
\end{array}
\]
We note the important aspect of this intuitionistic modelling is that any wff can be subject to attack. For example take the meta-level statement
\[\exists y(dRy)\]
this is false in the network. In the syntax it can be attacked by any $z$:
\[
\In(z)\to\neg\exists y(dRy)
\]

We can extend clauses $(A1)$, $(A2)$, $(B1)$, $(B2)$ to encompass any attacks on a wff $W$ by replacing ``$\In(x)$'' by ``$W$'', and $\{\text{``}\In(y)\text{''}\mid yRx\}$ by the wffs $\{Y\mid Y\tO W\}$ (i.e. the set of $Y$ that attack $W$ according to our network).
\end{remark}

\begin{example}\label{541-4D5}
Let us continue and analyse the higher level network of Figure~\ref{541-4F4}.

The network can be described as $(S,\mathbb{R})$, where $S=\{a,b,c,d\}$ and the relation $\mathbb{R}\subseteq (S\cup S^2)^2$ is:
\[
\{ (a,b), (a,c), (c,d), (a,(c,d)), ((a,b),d)
\}
\]
corresponding to
\[
\{ a\tO b, a\tO c, c\tO d, a\tO(c\tO d), (a\tO b)\tO d\}
\]
Because we have an attack predicate $R$ in the intuitionistic language we can express $\mathbb{R}$ using $R$.

To do this effectively we need to define the set $\mathbb{W}$ of elements (units) which can participate in the higher level attacks. These are formulas of intuitionistic logic built up from Figure~\ref{541-4F4}. Let us list them:
\[
\begin{array}{cl}
\multicolumn{2}{l}{\text{The set $\mathbb{W}$:}} \\
(W_1) & \In(a) \\
(W_2) & \In(b) \\
(W_3) & \In(c) \\
(W_4) & \In(d) \\
(W_5) & aRb \\
(W_6) & aRc \\
(W_7) & cRd \\
\end{array}
\]
The attack relation $\mathbb{R}$ can now be expressed in the object level of the language with $\In$ and $R$.
\begin{table}{ht}
\[
\begin{array}{c|c}
\text{Attack} & \text{representation}\\ \hline
a\tO b & (\In(a)\land aRb)\to\neg\In(b) \\ \hline
a\tO c &  (\In(a)\land aRc)\to\neg\In(c) \\ \hline
c\tO d &  (\In(c)\land cRd)\to\neg\In(d) \\ \hline
a\tO(c\tO d) &  \In(a)\to\neg cRd \\ \hline
(a\tO b)\tO d & aRb\to\neg\In(d)\\ \hline
\end{array}
\]
\caption{}\label{541-4F6}
\end{table}
The set $\mathbb{W}$ with elements $(W_1)$-$(W_7)$ together with Table~\ref{541-4F6} gives rise to the argumentation network of Figure~\ref{541-4F7}.
\begin{figure}
\centering
\ifx\JPicScale\undefined\def\JPicScale{1}\fi
\unitlength \JPicScale mm
\begin{tikzpicture}[x=\unitlength,y=\unitlength,inner sep=0pt]
\draw (11.25,30) node {$W_1=\In(a)$};
\draw (11.25,0) node {$W_3=\In(c)$};
\draw (0,15) node {$W_6=aRc$};
\draw (51.25,30) node {$W_2=\In(b)$};
\draw (51.25,0) node {$W_4=\In(d)$};
\draw (26.25,15) node {$W_7=cRd$};
\draw (26.25,42.5) node {$W_5=aRb$};
\draw [->>,>=triangle 45](11.25,27.5) -- (11.25,3.75);
\draw [->>,>=triangle 45](27.5,40) -- (48.75,3.75);
\draw [->>,>=triangle 45](12.5,27.5) -- (23.75,18.75);
\draw [->>,>=triangle 45](21.25,30) -- (41.25,30);
\draw [->>,>=triangle 45](21.25,0) -- (41.25,0);
\end{tikzpicture}
\caption{}\label{541-4F7}
\end{figure}
We need axioms to make this work. Axioms ($a1$), ($a2$), ($b1$) and ($b2$) now use $\mathbb{W}$ and Table~\ref{541-4F6}. They become the following ($\text{\it fa}1$), ($\text{\it fa}2$), ($\text{\it fb}1$), ($\text{\it fb}2$) and ($\text{\it fc}$). The construction uses Figure~\ref{541-4F7} and the attack relation of Table~\ref{541-4F6}, which we can call $(\mathbb{W},\mathbb{R}_1)$. We are now in the situation of Section~2 Definition~\ref{541-D1}.
\end{example}

The meaning of the clauses of Definition~\ref{541-D1} in terms of attacks and objects participating in attacks are as follows.

\begin{table}[ht]
\[
\begin{array}{@{\text{meaning of }}l@{:\quad}l}
a1 &\bigwedge_\text{all objects $W$}\bigl(W\text{ is in}\to(\bn\lor\bigwedge_\text{attackers $Y$ of $W$} Y\text{ is out})\bigr)
\\[0.3cm]
a2&
\bigwedge_\text{all objects $W$}\bigl(\bigwedge_\text{attackers $Y$ of $W$}\text{$Y$ is out}\to(\bn\lor W\text{ is in})\bigr)
\\[0.3cm]
b1 &
\bigwedge_\text{all objects $W$}\bigl(W\text{ is out}\to(\bn\lor\bigvee_\text{attackers $Y$ of $W$} Y\text{ is in})\bigr)
\\[0.3cm]
b2 &
\bigwedge_\text{all objects $W$}\bigl(\bigvee_\text{attackers $Y$ of $W$}\text{$Y$ is in}\to(\bn\lor W\text{ is out})\bigr)
\end{array}
\]
\caption{}\label{541-MT1}
\end{table}
If we follow the above meaning and use $(\mathbb{W},\mathbb{R}_1)$, we get the following:
\[
\begin{array}{@{\land}l}
\multicolumn{1}{l}{\bigl(\In(b)\to(\bn\lor\neg\In(a)\lor\neg aRb)\bigr)} \\
\bigl(\In(c)\to(\bn\lor\neg\In(a)\lor\neg aRc)\bigr) \\
\bigl(\In(d)\to(\bn\lor\neg\In(c)\lor\neg cRd)\bigr) \\
\bigl(cRd\to(\bn\lor\neg\In(a))\bigr) \\
\bigl(\In(d)\to(\bn\lor\neg aRb)\bigr) \\
\end{array}
\tag{\text{fa1}}
\]

\[
\begin{array}{@{\land}l}
\multicolumn{1}{l}{\bigl((\neg\In(a)\lor\neg aRb)\to(\bn\lor\In(b))\bigr)} \\
\bigl((\neg\In(a)\lor\neg aRc)\to(\bn\lor\In(c))\bigr) \\
\bigl((\neg\In(c)\lor\neg cRd)\to(\bn\lor\In(d))\bigr) \\
\bigl(\neg\In(a)\to(\bn\lor cRd)\bigr) \\
\bigl(\neg aRb\to(\bn\lor \In(d))\bigr) \\
\end{array}
\tag{\text{fa2}}
\]

\[
\begin{array}{@{\land}l}
\multicolumn{1}{l}{\bigl(\neg\In(b)\to(\bn\lor(\In(a)\land aRb))\bigr)} \\
\bigl(\neg\In(c)\to(\bn\lor(\In(a)\land aRc))\bigr) \\
\bigl(\neg\In(d)\to(\bn\lor(\In(c)\land cRd))\bigr) \\
\bigl(\neg cRd\to(\bn\lor\In(a))\bigr) \\
\bigl(\neg \In(d)\to(\bn\lor aRb)\bigr) \\
\end{array}
\tag{\text{fb1}}
\]

\[
\begin{array}{@{\land}l}
\multicolumn{1}{l}{\bigl((\In(a)\land aRb)\to(\bn\lor\neg\In(b))\bigr)} \\
\bigl((\In(a)\land aRc)\to(\bn\lor\neg\In(c))\bigr) \\
\bigl((\In(a)\land cRd)\to(\bn\lor\neg\In(d))\bigr) \\
\bigl(\In(a)\to(\bn\lor \neg cRd)\bigr) \\
\bigl( aRb\to(\bn\lor \neg \In(d))\bigr) \\
\end{array}
\tag{\text{fb2}}
\]

The reader will notice that we actually used conjunctive attacks in this example. Figure~\ref{541-4F7} really should be Figure~\ref{541-4F7a}. We shall discuss this formally in the next section~\ref{sec5}.
\begin{figure}
\centering

\ifx\JPicScale\undefined\def\JPicScale{1}\fi
\unitlength \JPicScale mm
\begin{tikzpicture}[x=\unitlength,y=\unitlength,inner sep=0pt]
\draw (12.5,30) node {$W_1$};
\draw [->>,>=triangle 45](12.5,27.5) -- (12.5,3.75);
\draw [->>,>=triangle 45](31.25,38.75) -- (50,3.75);
\draw [->>,>=triangle 45](13.75,27.5) -- (26.25,17.5);
\draw [->>,>=triangle 45](17.5,30) -- (47.5,30);
\draw [->>,>=triangle 45](17.5,0) -- (47.5,0);
\draw (28.75,41.25) node {$W_5$};
\draw (51.25,30) node {$W_2$};
\draw (51.25,0) node {$W_4$};
\draw (28.75,15) node {$W_7$};
\draw (12.5,0) node {$W_3$};
\draw (0,17.5) node {$W_6$};
\draw (3.75,15) -- (12.5,10);
\draw (28.75,38.75) -- (32.5,30);
\draw (28.75,12.5) -- (35,0);
\draw (6.25,13.75) .. controls (7.5,16.25) and (10,17.5) .. (12.5,16.25);
\draw (10,13.75) node {$\land$};
\draw (26.25,30) .. controls (25,32.5) and (27.5,36.25) .. (30,35);
\draw (28.75,32.5) node {$\land$};
\draw (28.75,0) .. controls (27.5,2.5) and (30,6.25) .. (32.5,5);
\draw (31.25,2.5) node {$\land$};
\end{tikzpicture}
\caption{Our notation for joint attacks is as in Figure~\ref{541-4F8}}\label{541-4F7a}
\end{figure}
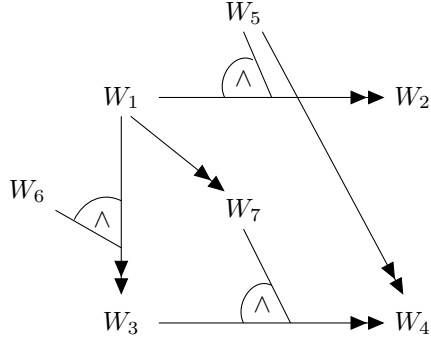

\begin{figure}
\centering
\ifx\JPicScale\undefined\def\JPicScale{1}\fi
\unitlength \JPicScale mm
\begin{tikzpicture}[x=\unitlength,y=\unitlength,inner sep=0pt]
\draw (0,30) node {$z_1,$};
\draw (15,30) node {$\dots,$};
\draw (30,30) node {$z_n$};
\draw (0,27.5) -- (15,15);
\draw (30,27.5) -- (15,15);
\draw [->>,>=triangle 45](15,15) -- (15,2);
\draw (15,0) node {$x$};
\draw (10,19) .. controls (14,22) and (16,22) .. (20,19);
\draw (15,19) node {$\land$};
\end{tikzpicture}
\caption{The reading is $\bigwedge_i\In(z_i)\to\neg\In(x)$}\label{541-4F8}
\end{figure}
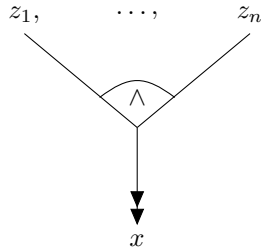

\section{Formal development of the intuitionistic model}\label{sec5}

We begin with further motivation for our meta-level view of argumentation networks. Let us start with a traditional network $(S,R)$ with $R\subseteq S\times S$, and let us focus on the following two types of activities associated with it.

\begin{enumerate}
\item {\bf Instantiation}

We instantiate elements $x$ of $S$ by arguments from some logical system $\BL$. This case is the ASPIC type of instantiation (see~\cite{541-15,541-16,541-17,541-18}). In all of these instantiations,  $\BL$ is {\em not} based on the language of $(S,R)$ itself. To be quite clear, we deal with wffs which talk about some application area where the arguments of $S$ come from, and are not about the geometry of $(S,R)$ itself.

There are arguments, however, which use the geometry of $(S,R)$. Such arguments are, for example, the formula $W(a)$, mentioned in section~\ref{sec1} and other arguments like
\begin{enumerate}
\item[$J(a,b):$] The arguments $a$ and $b$ attack the same arguments.
\end{enumerate}
Obviously they are not independent (think of two witnesses co-ordinating their testimony).
\[
J(a,b) = \forall x(aRx\leftrightarrow bRx)
\]
\item {\bf Abduction}

Abduction in argumentation is a meta-level operation on $(S,R)$ which can take two forms:
\begin{enumerate}
\item
\textit{Observation}: for example, we can ask if  there are extensions in which $x$ is in or $x$ is out, etc.
\item
\textit{Intervention}: for example, the instruction to delete $x$ from $S$ or to ensure $x$ is out (or in), or the instruction to delete (or ignore) all self-attacking arguments, or the like.
\end{enumerate}
\end{enumerate}
Observation is easy to implement in our intuitionistic framework. Given $\CA=(S,R)$, if we want to observe whether $\phi$ holds, we look for models of $\DA\cup\{\phi\}$. So if $x\in S$, we can take $\phi$ to be $\In(x)$ or $\neg\In(x)$, and so on. Since every model of $\DA$ is an extension, then the models of $\DA\cup\{\phi\}$ would yield extensions satisfying $\phi$, if they exist.

Intervention is more of a problem. How do we intervene and delete, or ignore, an argument $a\in S$? We can add to $S$ a new annihilator node $x_a$ with $x_aRa$. This will ensure that $a$ is out in $(S\cup\{x_a\},R\cup\{(x_a,a)\})$.

We may, however, want to intervene in a more subtle way. Say we want to delete every element which attacks every other element. We can do that by identifying such elements $a$ and adding an annihilator $x_a$ for each such $a$. This is not subtle at all. It does not give $a$ a chance to protect or defend itself. To make this case more concrete, suppose $a$ is an argument which attacks everything opposed to some ``Holy Book''. We might want to intervene and take $a$ out because we do not want such considerations to be involved. However, there may be arguments against such intervention. We may not want to conduct the discussion of whether to intervene or not in the meta-level. Furthermore, the discussion may involve object level considerations as well as meta-level considerations. For example, $a$ itself may contain counter argument intended to forestall intervention, or there may be an argument that if we were to take $a$ out we might get the wrong extensions, etc.

So how do we allow $(S,R)$ to include such arguments and let them attack? In the above example we must let $x_a$ be something like the formula below , which actually states the reason why we want to intervene and take $a$ out:
\[
x_a = \forall y(y\neq a\to aRy)
\]
and add $x_a$ to $S$ with $x_aRa$. In that case, we can allow $a$ to protect itself by saying $J(a)$:
\[
J(a) = \forall x(aRx\leftrightarrow xRx)
\]
i.e. $a$ defends itself by saying that it attacks self-attacking arguments (e.g. because the Holy Book is not tolerant of them).

The above discussion favours argumentation networks containing formulas of the predicate logic of $R$. We thus have as arguments the set $(S\cup W, R)$ with $R\subseteq(S\cup W)^2$ and $W$ is a set of wffs in the logic $\BLR$ based on the binary relation $R$.

If we think more along these lines, we see that we have a challenge for ASPIC itself. We know it is perfectly reasonable to discuss what arguments to allow into $(S,R)$. So let us do the ASPIC approach and use this discussion as the generator of the arguments for $(S,R)$. We thus have a loop here, how do we give semantics to the process?

Our intuitionistic logic can do this easily. The wffs of $\BLR$ are available in the language.  So:
\begin{itemize}
\item
If $\phi$ attacks $\psi$ we write $\phi\to\neg\psi$
\item
If $\phi$ attacks $x$ we write $\phi\to\neg\In(x)$
\item
If $x$ attacks $y$ we write $(\In(x)\land xRy)\to\neg\In(y)$
\item
If $x$ attacks $\phi$ we write $\In(x)\to\neg\phi$.
\end{itemize}
$\phi$ and $\psi$ can also be $xRy$ or $uRv$ and thus we can also have higher level attacks.

\begin{example}\label{541-EAug17-0}
Let us give a natural example of a network where a formula $\phi(R)$ attacks other arguments. Consider a network $(S,R)$ where $S=\upa{a,x,y,z}$ and where each of $x,y,z$ attacks $a$ (Figure~\ref{541-FAug17-1}). In this Figure, $a$ is attacked by all other arguments in the network (i.e. by everything else in $S$). This state of affairs is described by $\phi$ which can be used to protect $a$ (extending $R$). In this case $\phi$ can be seen as discrediting attacks as being part of an unfair, universal, attack on $a$.

\begin{figure}[ht]
\centering
\ifx\JPicScale\undefined\def\JPicScale{1}\fi
\unitlength \JPicScale mm
\begin{tikzpicture}[x=\unitlength,y=\unitlength,inner sep=0pt]
\draw (20,50) node {$\phi:\forall x(x\neq a\to xRa)$};
\draw (0,30) node {$x$};
\draw (20,30) node {$y$};
\draw (40,30) node {$z$};
\draw (20,0) node {$a$};
\draw [->>,>=triangle 45](17.5,47.5) -- (1.25,32.5);
\draw [->>,>=triangle 45](20,47.5) -- (20,32.5);
\draw [->>,>=triangle 45](22.5,47.5) -- (38.75,32.5);
\draw [->>,>=triangle 45](1.25,27.5) -- (18.75,2.5);
\draw [->>,>=triangle 45](20,27.5) -- (20,2.5);
\draw [->>,>=triangle 45](38.75,27.5) -- (22.5,3.75);
\end{tikzpicture}
\caption{}\label{541-FAug17-1}
\end{figure}
The situation is similar in Figure~\ref{541-FAug17-2}, where $S=\upa{a,b,x,y}$ and where $b$ attacks $a$, $x$, and $y$. $\phi$ protects $a$ from $b$ by attacking it on the grounds that it attacks everything.

Such arguments are very common in natural reasoning. We call them ``would-would'nt'' arguments. For example, the head of the opposition party might strongly attack the policies of the current government. The head of the current government might defend his policies by saying ``well the leader of the opposition \textit{would} say all of that \textit{wouldn't} he, he always attacks everything I say, because he is the leader of the opposition''.\footnote{Thanks to L. Rivlin for pointing out this case.}
\begin{figure}[ht]
\centering
\ifx\JPicScale\undefined\def\JPicScale{1}\fi
\unitlength \JPicScale mm
\begin{tikzpicture}[x=\unitlength,y=\unitlength,inner sep=0pt]
\draw (20,50) node {$\phi:\forall x(x\neq b\to bRx)$};
\draw (0,30) node {$x$};
\draw (20,30) node {$b$};
\draw (40,30) node {$y$};
\draw (20,0) node {$a$};
\draw [->>,>=triangle 45](20,47.5) -- (20,32.5);
\draw [->>,>=triangle 45](20,27.5) -- (20,2.5);
\draw [->>,>=triangle 45](22.5,30) -- (38.75,30);
\draw [->>,>=triangle 45](17.5,30) -- (2.5,30);
\end{tikzpicture}

\caption{}\label{541-FAug17-2}
\end{figure}
\end{example}

\begin{example}\label{541-5E1}
Consider Figure~\ref{541-5F2}.
\begin{figure}[ht]
\centering
\ifx\JPicScale\undefined\def\JPicScale{1}\fi
\unitlength \JPicScale mm
\begin{tikzpicture}[x=\unitlength,y=\unitlength,inner sep=0pt]
\draw (0,20) node {$a$};
\draw [->>,>=triangle 45](0,17) -- (0,3);
\draw (0,0) node {$\phi=\exists x\neg xRx$};
\end{tikzpicture}
\caption{}\label{541-5F2}
\end{figure}
This figure presents a problem. Since $a$ is not attacked, we have $a=\text{in}$ i.e. the value of $\In(a)$ is $(\top,\top)$. Thus $\phi$ must be out, i.e. $\phi=(\bot,\bot)$. But if $\exists x\neg xRx$ is false, we get that we must have $aRa$ and so how can $a$ be in? Then we must go for $\In(a)=\text{und}=(\bot,\top)$. So $\phi$ is also und, $\phi=(\bot,\top)$. So at the world $s$ we have $\In(a)\to\bn\lor\exists\neg xRx$. This holds.

Let us look at Figure~\ref{541-5F2} slightly differently. The domain of the Figure is $S=\{a\}$. So ``\,$\exists x\phi(x)$'' simply means $\phi(a)$. So $\exists x \neg xRx$ is just $\neg aRa$. Now $a$ attacking $\neg aRa$ is $\In(a)\to\neg\neg aRa$. This is also possible because the translation into intuitionistic logic says $\In(a)\to\bn\lor\neg\neg aRa$.

It looks like we need to accept that although a node $x$ may not be attacked geometrically, it may be attacked by some $\phi$ in the network by virtue of its meaning. The simplest example is Figure~\ref{541-5F3}.
\begin{figure}[ht]
\[
a\qquad \qquad aRa
\]
\caption{}\label{541-5F3}
\end{figure}
In Figure~\ref{541-5F3} the node $aRa$ does not attack the node $a$ geometrically, but they cannot both be in.

The situation, however, is still paradoxical. In Figure~\ref{541-5F2}, ``$\phi$''=$\exists x\neg xRx$ and ``$a$'' are consistent together. The problem arises because $a$ attacks $\phi$. We can consider the database $\{a,\exists x\neg xRx,a\to\neg \exists x\neg xRx \}$. This database describes the figure: for example if $a$ is in then it does not attack itself, and so $\exists x\neg xRx$ is also in. It is inconsistent. So perhaps we resolve the inconsistency by having the attack ``out'', i.e. we cancel the attack $a\to\neg\phi$.
\end{example}

The above example impresses upon us the need simply to ``run'' the translation into intuitionistic logic and see what it does. If we can accept what it does as reasonable we can proceed from there, otherwise we analyse why we find the result unacceptable and    seek an improvement. But our next step must be now to check the translation first.

\begin{remark}\label{541-5R5}
Let us analyse more closely the problems with the network of Figure~\ref{541-5F3}.
\begin{enumerate}
\item
The network has an unattacked element $aRa$ in it. Therefore $aRa$ is in. But being in, it changes the original network. Now $a$ attacks itself, so how can $a$ also be in?
\item
Suppose we adopt the remedy that we include only wffs $\phi$ that are true in $(S,R)$. If this is the case then $\phi=\top$, so how can $\phi$ ever be out, or how can $\phi$ ever be counter-attacked? So this `remedy' is not a good solution.
\item
So what is the source of the problem? The problem arises because when we construct the theory $\DA$, we look at $\CA=(S\cup W,R)$, with $R\subseteq(S\cup W)^2$, and get the complete extensions as model of $\DA$. In such models however, $R$ changes to $R'$ and maybe even $S$ changes to $S'$ (if some $\exists\phi(x)$ is in $W$ and is in). So we lose the correct representation theorem. So we must calculate the complete extension relative to $(S'\cup W,R')$. This means that $\DA$ must be changed. Can we do that implicitly in one move? We need to write $\DAp$ for $\CA'=(S',R')$, when we do not know what $\CA'$ is!
\item
Let us see how it is done for out example of Figure~\ref{541-5F3}, perhaps we can get some new ideas. Since $aRa$ is in, the network becomes $\CA'$, Figure~\ref{541-5F6}.
\begin{figure}
\centering
\ifx\JPicScale\undefined\def\JPicScale{1}\fi
\unitlength \JPicScale mm
\begin{tikzpicture}[x=\unitlength,y=\unitlength,inner sep=0pt]
\draw (0,10) node {$a$};
\draw (30,10) node {$aRa$};
\draw [->>,>=triangle 45](0,12.5) .. controls (8.75,27.5) and (16.25,0) .. (1.25,8.75);
\end{tikzpicture}
\caption{}\label{541-5F6}
\end{figure}
So we need to write $\DAp$ and not $\DA$. However, when we wrote the intuitionistic theory for Figure~\ref{541-5F3}, we wrote $\DA$, with, for example clause $(a2)$, by looking at Figure~\ref{541-5F3} and not by looking at Figure~\ref{541-5F6}. We wrote:
\[
\bigwedge_{\text{attackers $y$ of $a$}} y\text{ is out}\to(\bn\lor a\text{ is in})
\]
Since in Figure~\ref{541-5F3} there are no attackers of $a$ and the empty conjunction is $\top$, we got $\bn\lor\In(a)$. Had we looked at Figure~\ref{541-5F6} we would have had
\[
(\neg \In(a)\lor\neg aRa)\to(\bn\lor\In(a))
\]and the solution is $\In(a)=(\bot,\top)$.

The problem is that we do not know in general what $R'$ is going to be until we compute the extension, so we do not know what to write in order to compute the extension. So we ask, are we in a catch-22 situation? The answer is that there is a solution. Since we  have the predicate $xRy$ in the language, we do not need to look at the new Figure~\ref{541-5F6}. We simply adopt the view that every node $x\in S$ attacks every node $y\in S$ by the joint attack $\In(x)\land xRy$, independent of whether it is shown in $R$ or not. So for each node $x\in S$, to say that it is not attacked by other nodes, we write:
\[
\bigwedge_{y\in S}(\neg\In(y)\lor\neg yRx).
\]
To say, in addition, that $x$ is not attacked by any $\phi$ such that $\phi Rx$, we also add $\bigwedge_{\phi Rx}\neg \phi$ to the main conjunct. So $(a2)$ becomes $(a2^*)$:
\[
\begin{array}{c}
\biggl(\Bigl(\bigl(\bigwedge_{\phi Rx}\neg \phi\bigr)\land \bigl(\bigwedge_{yRx}(\neg\In(y)\lor\neg yRx)\bigr)\Bigr)\to(\bn\lor\In(x))\biggr)\\
\land \\
\bigwedge_{\psi\in W}\Bigl(\bigl(\bigwedge_{\phi R\psi}\neg \phi\bigr)\land\bigl(\bigwedge_{yR\psi}(\neg\In(y)\to(\bn\lor\psi)\bigr)
\Bigr)
\end{array}
\tag{$a2^*$}
\]
Any model of the theory $\DAp$ will tell us what $R$ is and that will make $(a2^*)$ behave correctly.

Let us do this for Figure~\ref{541-5F2}.

$(a1)$ becomes $(a1^*)$:
\[
\bigl(\In(a)\to(\bn\lor\neg\In(a)\lor\neg aRa) \bigr)\land\bigl(\exists x\neg xRx\to(\bn\lor\neg\In(a))\bigr)
\]
$(a2)$ becomes $(a2^*)$:
\[
\bigl((\neg\In(a)\lor\neg aRa)\to(\bn\lor\In(a))\bigr)\land\bigl(\neg\In(a)\to(\bn\lor\exists x\neg xRx) \bigr)
\]
$(b1)$ becomes $(b1^*)$:
\[
\bigl(\neg\In(a)\to(\bn\lor\neg\In(a)\lor aRa) \bigr)\land\bigl(\neg\exists x\neg xRx\to(\bn\lor\In(a))\bigr)
\]
$(b2)$ becomes $(b2^*)$:
\[
\bigl((\In(a)\land aRa)\to(\bn\lor\neg\In(a))\bigr)\land\bigl(\In(a)\to(\bn\lor\neg\exists x\neg xRx) \bigr)
\]
Let us now check what models this theory can have. Let us write the axioms from the point of view of the world $t$ (the actual world). $\bn$ is false at $t$ so we get:

\paragraph{Axioms $(a1^*)$ at $t$:}
\begin{itemize}
\item
$t\vDash\In(a) \Rightarrow t\vDash\neg\In(a)\lor\neg aRa$
\item
$t\vDash\exists x\neg xRx\Rightarrow t\vDash\neg\In(a)$
\end{itemize}
\end{enumerate}
\paragraph{Axioms $(a2^*)$ at $t$:}
\begin{itemize}
\item
$t\vDash\neg\In(a)\lor\neg aRa \Rightarrow t\vDash\In(a)$
\item
$t\vDash\neg\In(a)\Rightarrow t\vDash\exists x\neg xRx$
\end{itemize}
\paragraph{Axioms $(b1^*)$ at $t$:}
\begin{itemize}
\item
$t\vDash\neg\In(a) \Rightarrow t\vDash\In(a)\land aRa$
\item
$t\vDash\neg\exists x\neg xRx\Rightarrow t\vDash\In(a)$
\end{itemize}
\paragraph{Axioms $(b2^*)$ at $t$:}
\begin{itemize}
\item
$t\vDash\In(a)\land aRa \Rightarrow t\vDash\neg\In(a)$
\item
$t\vDash\In(a)\Rightarrow t\vDash\neg\exists x\neg xRx$
\end{itemize}

We now check what models this theory $\Delta =\{(a1^*),(a2^*),(b1^*),(b2^*)
\}$ can have. Suppose we have a model, what values can $\In(a)$ take?

\paragraph{Case 1.} $\In(a)=(\top,\top)$.\quad From $(a1^*)$ we get $t\vDash\neg aRa$, but this contradicts $(b2)$.

\paragraph{Case 2.} $\In(a)=(\bot,\top)$.\quad From $(b1^*)$ we get that $t\nvDash\neg\exists x \neg xRx$ and so $s\vDash \exists x \neg xRx$. Can $t\vDash \exists x \neg xRx$? From $(a1^*)$ we get that if $t\vDash \exists x \neg xRx$ we would have $t\vDash\neg\In(a)$, which contradicts $\In(a)=(\bot,\top)$. So $t\nvDash\exists x \neg xRx$ and so $\exists x \neg xRx=(\bot,\top)$.

Do we have a model? Do all axioms hold for $\In(a)=\exists x\neg xRx=(\bot,\top)$? We check and the answer is ``yes''.

\paragraph{Case 3.} $\In(a)=(\bot,\bot)$.\quad This contradicts $(a2^*)$.
\end{remark}

\begin{remark}[Intermediate evaluation]\label{541-5F7}

Let us see what we have so far. We have a procedure as follows:
\begin{enumerate}
\item
Given a network of the form $\CA=(S\cup W),R')$ with $S$ a set of nodes and $W$ a set of wff in the language $\BL(R)$ of a binary relation $R$ and $R'\subseteq(S\cup W)^2$, we construct a theory $\DAp$ on the network $\CA'$ using clauses $(a1^*)$, $(a2^*)$, $(b1^*)$, and $(b2^*)$
where these clauses are constructing according to the meanings given by Table~\ref{541-MT1} of Example~\ref{541-4D5}.
\item
We can define the complete extensions for $\CA'$ as the models of $\DAp$. This definition  is motivated by our discussions so far. We need to make some soundness checks and if we can, add a representation theorem reducing this proposed definition to some familiar terms.
\end{enumerate}
Let us do a quick soundness/plausibility check for this definition. Suppose we have no $W$, i.e. $W=\varnothing$. Then our network is $\CA=(S,R)$ with $R\subseteq S\times S$.\footnote{We use the letter $\CA$ (and not $\CA'$) when referring to $(S,R)$ in order to indicate that we want to consider $(S,R)$ as a traditional network. When we wish to indicate $(S,R)$ is  a network with $W=\varnothing$ then we refer to it using $\CA'$.} We know how to find the complete extensions for $\CA$, as explained in Section~\ref{sec2} in Definition~\ref{541-D1} and we have the Representation Theorem~\ref{541-T2} showing the correctness of the translation, into propositional $\BG_3$, of clauses $(a1)$, $(a2)$, $(b1)$ and $(b2)$.
These clauses are propositional and do not use any predicates. But now we can also look at $(S,R)$ as an argument network with $W$ empty and so we need to look at $\CA$ as $\CA'=(S,R)$. This is to be considered as another network which is translated into predicate intutitionistic logic $\pG$, with the predicates $\In$ and $R$ in the language. This translation uses the clauses $(a1^*)$, $(a2^*)$, $(b1^*)$ and $(b2^*)$ 
 which follow the meanings given by Table~\ref{541-MT1}.

Do we get the same extensions? The answer is ``yes''. We will not give a proof but an explanation. Consider for example clause $(a2)$. It says:
\[
\bigwedge_{x}\bigl(\bigvee_{yRx} y\to(\bn\lor\neg x)\bigr) \tag{a2}
\]
In predicate logic this clause becomes:
\[
\forall x\bigl(\exists y(yRx\land\In(y))\to(\bn\lor\neg \In(x))\bigr)\tag{A2}
\]
we can agree that the two clauses say the same thing and this observation was indeed exploited in Section~\ref{sec4} Remark~\ref{541-4R1} and Theorem~\ref{541-4T2}. The clause $(a2^*)$ for out case of $\CA'$ becomes:
\[
\forall x\bigl( \bigvee_{\text{$y$ attacks $x$}}\text{$y$ is in}\to(\bn\lor\text{$x$ is out})\bigr). \tag{$a2^*$}
\]
Now since $W=\varnothing$,  the attacks of $x$ are just $\{y\mid yRx\}$ and therefore $(a2^*)$ and $(A2)$ are the same.

So the definition has passed this quick soundness check.
\end{remark}

\begin{remark}\label{541-RAug1}\
\begin{enumerate}
\item
Let $\CA=(S_0,R_0)$ be a network with $m$ elements. We can assume for simplicity that $S_0=\upa{1,2,\dots,n}$ and $R\subseteq S_0^2$. Let $\BL(R,\In,=)$ be an intuitionistic language with unary $\In$ and binary $R$ and the equality symbol $=$. We shall use $\BL(R,\In,=)$ to describe $(S_0,R_0)$ completely in $\pG$ intuitionistic logic by a theory which we call $\OA$.
\item
$\OA$ has the following clauses:
\begin{enumerate}
\item
Axioms for equality $=$, together with the decidedness axioms:
\[
\forall xy(x=y\lor\neg x=y)
\]
\[
\forall xy(xRy\lor\neg xRy)
\]
\item
\[
\exists x_1,\dots,x_n\Bigl(\forall y\bigl(\bigvee_{i=1}^n y=x_i\bigr)\land \PA\Bigr)
\]
where we set $\PA=\bigwedge_{(i,j)\in R_0}x_iRx_j$.
\end{enumerate}
\item
Let $(A1)$, $(A2)$, $(B1)$ and $(B2)$
 be as in Remark~\ref{541-4R1} and let $\DA$ be the theory $\upa{(A1), (A2), (B1), (B2)
 }$.
\item
The following holds:
\begin{itemize}
\item
$h$ is a model of $\DA\cup\OA$ iff $\lambda_h$ is a complete extension of $(S_0, R_0)$.
\end{itemize}
\end{enumerate}
\end{remark}

\begin{remark}\label{541-RAug2}
Continuing the discussion of Remark~\ref{541-RAug1}
we note that we can regard the network $(S_0,R_0)$ of Remark~\ref{541-RAug1} as a network $\CA'=(S_0\cup\{\OA\},R')$ in the sense of Remark~\ref{541-5F7}, with $R'=R_0$. This network contains $\OA$ unattacked. Therefore $\OA$ holds in any model $h$, and the axioms of $\DA$ make sure the model is a complete extension.

So $\OA$ describes the network and $\DA$ forces the models to be complete extensions of the network which $\OA$ describes.

Viewed in this way, we notice that $\PA=\bigwedge_{(i,j)\in R_0}x_iRx_j$ is embedded in the wff $\OA$ and plays the role of describing the attack relation $R_0$ of $\CA$. $\PA$ does this descriptive job extensionally, in a traditional manner, by listing what attacks what. In this set up, however, we can allow $\PA$ to describe $R$ axiomatically, not by listing the members of $R$ but by writing axioms for $R$. For example we can take an axiom $\phi$ to be:
\[
\phi\quad=\quad \neg\exists x xRx\land\exists y\forall z\neg zRy.
\]
$\phi$ says that $R$ has no self attacking elements and at least one element that is not attacked. By letting $\PA=\phi$ we are looking at a new type of generic networks of the form $(S_0,\phi(R_0))$, where $R_0$ is restricted axiomatically by $\phi$.\footnote{Define the concept of an {\em axiomatic argumentation network} to be of the form $(S_0,\phi(R_0))$ with $\phi$ a wff describing properties of $R$. We consider this further in Section~\ref{sec6}.}

So to give an example, consider $S_0$ as a set of two elements $\upa{a,b}$. $\phi$ allows for the following options for $R_0$ (this is Example~\ref{541-EAug3} of Section~\ref{sec6} below):
\begin{enumerate}
\item[($i$)]
$R^1_0=\varnothing$
\item[($ii$)]
$R^2_0=\upa{a\tO b}$
\item[($iii$)]
$R^3_0=\upa{b\tO a}$
\end{enumerate}
The complete extensions are
\begin{enumerate}
\item[($i^*$)]
a=b=\text{in}
\item[($ii^*$)]
a=\text{in}, b=\text{out}
\item[($iii^*$)]
a=\text{out}, b=\text{in}
\end{enumerate}
Note that we can describe $\phi$ equivalently in this case as
\[
\phi\quad=\quad \bigvee_{i=1}^3(R_0^i\land\bigwedge_{j\neq i}\neg R_0^j).
\]

We further note that we can allow $\phi$ to contain additional predicates besides $R$ which help describe what kind of $R$ we want. These additional predicates may come from an application context for $R$ and may also include the predicate $\In$. We shall say more on this in Section~\ref{sec6}.
\end{remark}

\section{Axiomatic argumentation frames (AAF)}\label{sec6}

This section continues the ideas of Remark~\ref{541-RAug1}. We now moticate the new concept of \emph{Axiomatic Argumentation Frames} (AAF) and give formal definitions.

\begin{example}\label{541-EAug1}
Suppose we are given an argumentation network $\CA=(S,R)$ which is too large in the sense that $S$ has too many members, and that we wish to select a subset of $\CA$. For example, suppose we need to give a presentation and, for practical reasons (e.g. time, intelligibility etc.), we wish to select only a subset $S'$ of $S$ to appear in this presentation.\footnote{But we would wish the presentation to remain coherent view that is as self-contained as possible.} Let us suppose we have calculated that our presentation has time to discuss exactly $k$ arguments. So we need to pick a suitable $k$ sized subset $S'$ from $S$.

But not any $S'$ will do, as we might have further conditions on what $S'$ should look like: for example we might wish it to appear to present a coherent position say by having a nontrivial complete extension; or we might wish it to reflect $S$ in some way; or we may have some other external criteria for selecting $S'$ (such as degree of relevance to the topic). When choosing $S'$ we will check whether these conditions are met by $(S',R{\upharpoonright} S')$. Ultimately, we will need to write a wff $\Psi$ describing these conditions in predicate logic and then consider the sets $(S_j,R{\upharpoonright} S_j)$ which satisfy $\Psi$.
\end{example}

\begin{definition}\label{541-DAug2}\
\begin{enumerate}
\item
Let $\BL$ be a predicate language containing a binary predicate $R$ and a unary predicate $S$. Let $\Psi$ be a consistent classical theory of $\BL$. Let $\BM=\BM(\Psi)$ be a model of $\Psi$. Let $S_{\BM}$ be the extension of $S$ in $\BM$ and let $R_{\BM}$ be the binary extension of the relation of $\BM$. We ignore the extensions of all other predicates of $\BL$ which are involved in $\BM$. We view them as parameters helping to write the axiomatic properties of $S$ and $R$ (via the theory $\Psi$).
\item
We can view $\CA_{\BM}=(S_{\BM},R_{\BM}{\upharpoonright}S_{\BM})$ as a traditional argumentation network with complete extensions $E^{\BM}_1,\dots,E^{\BM}_{r(\BM)}$.
\item
Let $S_0$ be a finite set and let $\{{\bf a}\mid a\in S_0\}$ be the set of names in $\BL$ for elements of $S_0$. Let $\phi(S_0)$ be the formula
\[
\phi(S_0)\quad =\quad \forall x\Bigl( \bigvee_{a\in S_0} x={\bf a}\Bigl)\land \bigwedge_{\substack{
a,b\in S_0
\\
a\neq b
}} {\bf a}\neq {\bf b}.
\]
Let $\Psi(S_0)=\Psi\land\phi(S_0)$.
\item
We define the notion of the axiomatic argumentation network $\CA_0=(S_0,\Psi(S_0))$ to be the family of networks obtained from models $\BM$ of $\Psi(S_0)$ as defined in item (1) above. We also define the extensions of $\CA_0$ as all the extensions $E^{\BM}_{i}$ of all possible models  $\BM$ of $\Psi(S_0)$. Note that because of the axiom $\Psi(S_0)$, all networks $(S_{\BM},R_{\BM}{\upharpoonright}S_{\BM})$ satisfy $S_{\BM}=S_0$.
\end{enumerate}
\end{definition}

\begin{example}\label{541-EAug3}
Let $S_0=\{a,b\}$ and let $\Psi(S_0)$ be
\[
\begin{array}{rcl}
\Psi(S_0) &=&\forall x(x={\bf a}\lor x={\bf b})\land {\bf a}\neq {\bf b}
\land \exists x\forall y\neg yRx\land\neg {\bf a}R {\bf a}\land \neg {\bf b}R {\bf b}
\end{array}
\]
Then $(S_0,\Psi(S_0))$ has the following extensions:
\begin{itemize}
\item
$a=b=\text{in}$
\item
$a=\text{in}$, $b=\text{out}$
\item
$a=\text{out}$, $b=\text{in}$
\end{itemize}
\end{example}

\begin{example}\label{541-EAug4}
Following Example~\ref{541-EAug1}, assume we have 5 possible arguments of which we need to select 2 good ones to discuss and assume $R$ is as in Figure~\ref{541-FAug5} and that we would prefer to discuss any argument $a_n$ over $a_{n+1}$ if we could.\footnote{That is, we rank $a_1$ as more worthy of discussion than $a_2$, $a_2$ as more worthy than $a_3$ etc.}
\begin{figure}[ht]
\centering
\ifx\JPicScale\undefined\def\JPicScale{1}\fi
\unitlength \JPicScale mm
\begin{tikzpicture}[x=\unitlength,y=\unitlength,inner sep=0pt]
\draw (0,20) node {$a_4$};
\draw (40,20) node {$a_1$};
\draw (20,10) node {$a_3$};
\draw (0,0) node {$a_5$};
\draw (40,0) node {$a_2$};
\draw [->>,>=triangle 45](3.75,18.75) -- (17.5,11.25);
\draw [->>,>=triangle 45](23.75,12.5) -- (35,17.5);
\draw [->>,>=triangle 45](40,17.5) -- (40,2.5);
\draw [->>,>=triangle 45](35,1.25) -- (23.75,7.5);
\draw [->>,>=triangle 45](17.5,7.5) -- (3.75,1.25);
\draw [->>,>=triangle 45](0,2.5) -- (0,17.5);
\end{tikzpicture}
\caption{}\label{541-FAug5}
\end{figure}

We ask what we can take as $(S_j,R{\upharpoonright} S_j)$?
\begin{enumerate}
\item
If $S_j=\upa{a_1,a_2}$, then because $a_1\tO a_2$, we do not get enough choice.
\item
If we take $S_j=\upa{a_1,a_2,a_3}$, we do not get a proper extension.
\item
If we take  $S_j=\upa{a_1,a_2,a_3,a_4}$ we get the extension $E=\upa{a_1,a_4}$.
\end{enumerate}
This is our answer.
\end{example}

\begin{example}[Disjunctive Networks]\label{541-EAug5}

We show that networks with disjunctive attacks (see~\cite{541-19}) can be easily represented as Axiomatic Networks.

Networks with disjunctive attacks have the form $(S,\rhod)$, where $\rhod\subseteq S\times(2^S-\varnothing)$. A typical disjunctive attack of $z$ on $Y=\upa{y_1,\dots,y_k}$, $z\tO Y$, is drawn in Figure~\ref{541-FAug6}.
\begin{figure}[ht]
\centering
\ifx\JPicScale\undefined\def\JPicScale{1}\fi
\unitlength \JPicScale mm
\begin{tikzpicture}[x=\unitlength,y=\unitlength,inner sep=0pt]
\draw (20,40) node {$z$};
\draw (0,0) node {$y_1$};
\draw (40,0) node {$y_k$};
\draw (20,0) node {\dots};
\draw (20,37.5) -- (20,20);
\draw [->>,>=triangle 45](20,20) -- (2.5,2.5);
\draw [->>,>=triangle 45](20,20) -- (37.5,2.5);
\end{tikzpicture}
\caption{}\label{541-FAug6}
\end{figure}
One of the options for the semantics for $(S,\rhod)$ is to understand Figure~\ref{541-FAug6} as saying:
\begin{itemize}
\item
if $z=\text{in}$ then at least one of $y_i\in Y$ must be out.
\end{itemize}
The translation of the above statement for each attack $z\tO Y$ in $\rhod$ is $\Psi(\rhod,R)$:
\[
\Psi(\rhod,R) = \bigwedge_{z\rhod Y}\bigvee_{y\in Y}zRy
\]
where $R$ is a new symbol for a relation $R\subseteq S\times S$. Therefore, according to this semantics (called the \textit{open reductionist approach} in~\cite{541-19}), the extensions of $(S,\rhod)$ are the same as the extensions of $(S,\Psi(R))$ as an axiomatic network (see~\cite[Definition 2.8, Theorem 2.9]{541-19}.
\end{example}

\begin{example}[Conjunctive Networks]\label{541-EAug7}
Conjunctive networks have the form $(S,\rhoc)$, where $\rhoc\subseteq(2^S-\varnothing)\times S$. A typical conjunctive attack of  $Y=\upa{y_1,\dots,y_k}$ on $z$, $Y\tO z$, is drawn in Figure~\ref{541-FAug8}.
\begin{figure}[ht]
\centering
\ifx\JPicScale\undefined\def\JPicScale{1}\fi
\unitlength \JPicScale mm
\begin{tikzpicture}[x=\unitlength,y=\unitlength,inner sep=0pt]
\draw (20,0) node {$z$};
\draw (0,40) node {$y_1$};
\draw (40,40) node {$y_k$};
\draw (20,40) node {\dots};
\draw [->>,>=triangle 45](20,20) -- (20,2.5);
\draw (37.5,37.5) -- (20,20);
\draw (2.5,37.5) -- (20,20);
\end{tikzpicture}
\caption{}\label{541-FAug8}
\end{figure}
The condition for conjunctive attacks is
\[
\bigwedge_i y_i=\text{in}\text{ implies } z=\text{out}.
\]
Conjunctive attacks can be implemented using the relation $R\subseteq S\times S$ of point to point attacks, with the help of additional auxiliary points. For each $Y\rhoc z$, let the auxiliary points be
\[
H(Y,z)\quad =\quad \upa{\alpha(y,Y,z)\mid y\in Y}\cup\upa{\beta(Y,z)}
\]
Consider Figure~\ref{541-FAug9}, (this construction is from~\cite{541-20}).
\begin{figure}[ht]
\centering
\ifx\JPicScale\undefined\def\JPicScale{1}\fi
\unitlength \JPicScale mm
\begin{tikzpicture}[x=\unitlength,y=\unitlength,inner sep=0pt]
\draw (20,20) node {$\beta(Y,z)$};
\draw (0,60) node {$y_1$};
\draw (40,60) node {$y_k$};
\draw (20,60) node {\dots};
\draw (20,0) node {$z$};
\draw (0,40) node {$\alpha(y_1,Y,z)$};
\draw (40,40) node {$\alpha(y_k,Y,z)$};
\draw (20,40) node {\dots};
\draw [->>,>=triangle 45](0,57.5) -- (0,43.75);
\draw [->>,>=triangle 45](40,57.5) -- (40,42.5);
\draw [->>,>=triangle 45](1.25,37.5) -- (18.75,22.5);
\draw [->>,>=triangle 45](38.75,37.5) -- (21.25,22.5);
\draw [->>,>=triangle 45](20,17.5) -- (20,2.5);
\end{tikzpicture}
\caption{}\label{541-FAug9}
\end{figure}
Clearly, given $Y$ and $z$ such that $Y\tO z$, if we add the points $H(Y,z)$ and express that they are all distinct and unique to $(Y,z)$, then the point $\beta(Y,z)$ fills the role of $Y$ and the attack $\beta(Y,z)\tO z$ takes the role of $Y\rhoc z$.

It is also clear that in Figure~\ref{541-FAug9}, $z$ is out only if all $y\in Y$ are in. So if we are given a conjunctive network $(S_0,\rho_0)$,\footnote{We drop the subscript $\land$ to avoid a proliferation of symbols.} with $\rho_0\subseteq S_0\times(2^{S_0}-\varnothing)$ we can pass to a new network $(S_1,R_1)$ where: $S_1=S_0\cup\bigcup_{Y\rho_0 z}H(Y,z)$; and $R_1$ is defined as in Figure~\ref{541-FAug9} for each $H(Y,z)$. $(S_1,R_1)$ will `implement' $(S_0,\rho_0)$ in the sense that the complete extensions $E_1$ of $(S_1,R_1)$ generate, by the projection $E_0=E_1\cap S_0$, exactly the complete extensions $E_0$ of $(S_0\rho_0)$. In other words
\begin{itemize}
\item
$E_0$ is an extension of $(S_0,\rho_0)$ iff for some extension $E_1$ of $(S_1,R_1)$ we have $E_0=E_1\cap S_0$. \end{itemize}

Using the above route we can implement conjunctive networks as axiomatic networks. We need to define a language $\BL$ and a $\Psi$ in $\BL$ which allows us to identify, using $\Psi$, the network $(S_0,\rho_0)$. So now let us define $\rho$ by a formula in an axiomatic argumentation network. The axiom $\Psi$ we take simply expresses the situation of Figure~\ref{541-FAug9}.

We need a language $\BL(S,R,\rho)$ and we define $\rho$ using $R$. We need some auxiliary notation.

Let $S_0$ be a finite set and assume we have names $\bs$ in $\BL$ for every $s\in S_0$. Let $\BS_0=\upa{\bs\mid s\in S_0}$. Given a set of variables $V$, let $D(V)$ be the formula:
\[
D(V)\quad =\quad \bigwedge_{\substack{x\neq y\\ x,y\in V}}\neg(x=y)\land\forall x\Bigl(\bigvee_{x\in V}z=x \Bigr)
\]
Let $\exists_V$ be the prefix quantifier (for $V$ finite, i.e. $V=\upa{v_1,\dots,v_n}$) then:
\[
\exists_V\phi\quad =_\text{def}\quad \exists v_1\dots v_n\Bigl(D(V)\land \phi \Bigr).
\]
Now suppose we are given a conjunctive network $\CA_0=(S_0,\rho_0)$, with $\rho_0\subseteq S_0\times(2^{S_0}-\varnothing)$. Let $\BL(S,R,\rho,\BS_0)$ be a predicate language with unary $S$, binary $R$ and $\rho$ and the set of constants $\BS_0=\upa{\bs\mid s\in S_0}$. Let $V_0$ be the set of distinct variables
\[
V_0\quad=\quad\bigcup_{Y\rho_0 z}\upa{\alpha(y,Y,z),\beta(Y,z)\mid y\in Y}.
\]
Let $\Psi(S,R,\rho,\BS_0)$, or $\Psi$ for short, be the formula
\[
\begin{array}{lcrll}
\Psi&=&
\exists_{V_0\cup\BS_0}\biggl[&\forall z\Bigl(S(z)\leftrightarrow \bigvee\limits_{s\in S_0} z=\bs \Bigr) \\
&&& \land \bigwedge\limits_{\substack{Y\rho_0 z  \\ Y\subseteq S_0\\ z\in S_0}}\Phi(Y,z)
\\
&&& \land
\bigwedge\limits_{\substack{Y\rho_0 s \\ Y\subseteq S_0 \\ s\in S_0}} \bigwedge\limits_{y\in Y}\Bigl(yR\alpha(y,Y,s)\land \alpha(y,Y,s)R\beta(Y,s)\land\beta(Y,s)R\bs
\Bigr) & \biggr]
\end{array}
\]
where $\Phi(Y,z)$ is the following wff:
\[
\begin{array}{lcrll}
\Phi(Y,z)&=&
Y\rho z \leftrightarrow \exists y_1\dots y_k\biggl(&\bigwedge\limits_{i\neq j} u_i\neq y_j\land \bigwedge\limits_{j}S(y_j)\land S(z) \\
&&& \land \exists u\exists x_1\dots x_k\Bigl(\bigwedge\limits_j\neg S(x_j)\land\bigwedge\limits_{i\neq j}x_i\neq x_j\land\neg S(u) \\
&&&\multicolumn{1}{r}{\land\bigwedge\limits_{j}(y_jRx_j\land x_j Ru)\land uRz\Bigr)}
& \biggr)
\end{array}
\]
Let $\BM$ be a model of $\Psi$. The original $(S_0,\rho_0)$ can be obtained from this model by letting $S_0$ be the extension $S_{\BM}$ of $S$ in the model. Also the relation $Y\rho_0 s$, for $Y\subseteq S_0$ and $s\in S_0$ can be defined as the extension $S_{\BM}$ of $\rho$.
\end{example}

\begin{example}[Abstract Dialectical Frameworks]\label{541-EAug10}

We now implement Abstract Dialectical Frameworks (ADF),~\cite{541-21}, in our Axiomatic approach (AAF).

We first describe the first order version of ADF. Consider Figure~\ref{541-FAug11}. This figure describes the basic configuration of an attack on a point $x$ in $(S,R)$.
\begin{figure}[ht]
\centering
\ifx\JPicScale\undefined\def\JPicScale{1}\fi
\unitlength \JPicScale mm
\begin{tikzpicture}[x=\unitlength,y=\unitlength,inner sep=0pt]
\draw (20,0) node {$x$};
\draw (0,20) node {$y_1$};
\draw (40,20) node {$y_k$};
\draw (20,20) node {\dots};
\draw [->>,>=triangle 45](1.25,17.5) -- (18.75,2.5);
\draw [->>,>=triangle 45](38.75,17.5) -- (21.25,2.5);
\end{tikzpicture}
\caption{}\label{541-FAug11}
\end{figure}
$\upa{y_1,\dots,y_k}$ are all the attackers of $x$. In traditional Dung semantics we have that:
\begin{itemize}
\item
$x=\text{in}$ iff all $y_i$  are out
\end{itemize}
the ADF semantics gives a local condition $F(x)$ on $\upa{y_1,\dots,y_k}$, saying
\begin{itemize}
\item
$x=\text{in}$ iff one of certain in/out distributions on $\upa{y_1,\dots,y_k}$ holds.
\end{itemize}
Assuming that any of these distribution can be expressed in classical propositional logic on $\upa{y_1,\dots,y_k}$, we can write $F(x)$ as a wff of the form
\[
F(x)\quad =\quad \bigvee_{\epsilon\in\Delta_x}\bigwedge_i y_i =\epsilon(i)
\]
where $\Delta_x$ is a set of vectors $\epsilon:\upa{1,\dots,k}\mapsto\upa{0,1}$ where ``$y=1$'' means ``$y=\text{in}$'' and ``$y=0$'' means ``$y=\text{out}$''.

We can also write $F(x)$ as $\bigvee_{\epsilon\in\Delta_x}\bigwedge_i y_i^{\epsilon(i)}$, where $y^1=y$ and $y^0=\neg y$. So we have:
\begin{itemize}
\item
$x=\text{in}$ iff $\bigvee_{\epsilon\in\Delta_x}\bigwedge_i y_i =\epsilon(i)$
\end{itemize}
or
\begin{itemize}
\item
$F_S = \bigwedge_{x\in S} \bigl(x\leftrightarrow F(x)\bigr)=\top$
\end{itemize}
The latter formulation of $F$ is mentioned as formula (1) on the third page of~\cite{541-20}. To conform with the notation of the present paper we should write $F$ as $F'$
\begin{itemize}
\item
$F'(x) = \bigvee_{\epsilon\in\Delta_x}\bigwedge_i \In(y_i)^{\epsilon(i)}$,
\end{itemize}
and the full wff is
\begin{itemize}
\item
$F'_S = \bigwedge_{x\in S} \bigl(\In(x)\leftrightarrow F'(x)\bigr)=\top$.
\end{itemize}

The reader should note that the formula $F'_S$ cannot be considered an argument in our sense of Section~\ref{sec5}, because $F'$ is not a wff about $R$. We cannot add it to the specification of the model because it contradicts the traditional Dung requirement:
\[
\bigwedge_{x\in S} \bigl(\In(x)\leftrightarrow \bigwedge_{yRx}\neg\In(y)\bigr)
\]
which is part of our specification. So our only recourse, if we want to implement (ADF) in (AAF), is to add auxiliary points and follow similar route as taken in Example~\ref{541-EAug7}.

To simplify the presentation and to save on complicated wffs let us present what is to be done for the configuration of Figure~\ref{541-FAug12}.
\begin{figure}[ht]
\centering
\ifx\JPicScale\undefined\def\JPicScale{1}\fi
\unitlength \JPicScale mm
\begin{tikzpicture}[x=\unitlength,y=\unitlength,inner sep=0pt]
\draw (0,30) node {$a$};
\draw (20,30) node {$b$};
\draw (40,30) node {$c$};
\draw (20,0) node {$x=((a\land\neg b)\lor c)$};
\draw [->>,>=triangle 45](1.25,26.25) -- (16.25,3.75);
\draw [->>,>=triangle 45](20,26.25) -- (20,3.75);
\draw [->>,>=triangle 45](38.75,26.25) -- (23.75,3.75);
\end{tikzpicture}
\caption{}\label{541-FAug12}
\end{figure}
Let $Y=\upa{a,b,c}$ be the set of all attackers of $x$. The condition for $x$ to be ``in'' is that either both $a=\text{in}$ and $b=\text{out}$ or else $c=\text{in}$.

To implement this figure we add the auxiliary points $\beta(x,Y)$, $\gamma(x,a,b,Y)$ and $\gamma(x,c,Y)$.  $\beta(x,Y)$ stands for $\neg x$ and $\gamma(x,a,b,Y)$ stands for the disjunct $a\land\neg b$ and $\gamma(x,c,Y)$ stands for the disjunct $c$.

We also need, for each $\neg z$ at each disjunct, a point to stand for $z$. In our case we have the disjunct $a\land\neg b$ and $\neg z=\neg b$. So we need a point $\delta(\gamma(x,a,b,Y),b)$. We now create Figure~\ref{541-FAug13}.

We will not need some of these additional points, as you will see in the Figure~\ref{541-FAug13} we need only deal with disjuncts which have negated atoms in them.
\begin{figure}[ht]
\centering
\ifx\JPicScale\undefined\def\JPicScale{1}\fi
\unitlength \JPicScale mm
\begin{tikzpicture}[x=\unitlength,y=\unitlength,inner sep=0pt]
\draw (0,65) node {$a$};
\draw (30,65) node {$b$};
\draw (60,65) node {$c$};
\draw (30,45) node {$\delta(\gamma(a,b,Y),b))$};
\draw (30,20) node {$\beta(x,Y)$};
\draw (30,0) node {$x$};
\draw [->>,>=triangle 45](58.75,61.25) -- (32.5,23.75);
\draw [->>,>=triangle 45](30,61.25) -- (30,47.5);
\draw [->>,>=triangle 45](1.25,61.25) -- (27.5,23.75);
\draw (30,42.5) -- (21.25,32.5);
\draw [->>,>=triangle 45](30,16.25) -- (30,2.5);
\draw (26.5,38.4) [rotate=90.0]  arc (-43.4:39.6:8.12 and 7.19);
\draw (21,38) node {$\land$};
\end{tikzpicture}
\caption{}\label{541-FAug13}
\end{figure}

\end{example}

\section{Conclusion and discussion}

We have translated argumentation networks of the form $\CA=(S,R)$ into intuitionistic logic, $\tau:A\mapsto\tau A$, in such a way that the following holds:
\begin{enumerate}
\item
Complete extensions $E$ of $\CA$ correspond to models $h_E$ of $\tau(\CA)$ and vice versa.
\item
The attack relation in $\CA$, $x\tO y$, is interpreted intuitionistically as $x\rightarrow \neg y$ (where $\rightarrow$ and $\neg$ are intuitionistic propositional connectives).
\item
As a byproduct of this representation we obtained a coherent way of adding, to any $\CA$, statements $\phi(R)$ about $R\in \CA$ to serve also as arguments. Moreover, our intuitionistic interpretation $\tau$ allowed us to give these statements semantics (see Examples~\ref{541-EAug17-0} and~\ref{541-5E1}).
\item
The reader might have some doubts about the idea of such an interpretation. Let us make several additional points.
\begin{enumerate}
\item
Members of the argumentation community take pride in the added value of the argumentation point of view. They stress the programme of interpreting, for example, non-monotonic systems into it. They should note that any support or objections to interpreting Abstract Argumentation {\em out} into systems $Y$, may equally be applied to interpreting any system $X$ {\em into} an argumentation framework.
\item
Whenever we give a good interpretation of a system $X$ into a system $Y$, our understanding of both systems stands to benefit. System $X$ may get additional meaning as part of $Y$ and the formal understanding of system $Y$ may be enhanced through its hosting of system $X$. One such example was the translation of modal logic into classical logic via the translation $\tau$:
\[
\tau(x\vDash \Box A) = \forall y (xRy\rightarrow \tau(y\vDash A))
\]
Through this translation, propositional modal logic obtained another semantics, and formal facts about fragments of first order classical logic became apparent. For example, through the observation (e.g. in~\cite{541-24}) that modal logic needs to use only 2 bounded, i.e. guarded, variables in this translation, we obtain decidability results for guarded quantifiers in classical first order logic, (see~\cite{541-25,541-26}).
\end{enumerate}
\item
We would like to compare the translation of this paper with that of~\cite{541-22}.

The paper of Dvorak {\em et. al.}, \cite{541-22}, is an example of a general translation. The formal mathematical language is second order monadic first order logic. This can serve as the modelling language for the majority, if not all, of the varieties of argumentation networks. It is intended by the authors to be to argumentation as ALGOL is to algorithms. It is an exact mathematical logic language strong enough to express whatever you want to say about argumentation networks. To quote the authors of~\cite{541-22}:
\begin{quote}
Begin quote.\\
We propose the formalism of monadic second order logic (MSO) as
a unifying framework for representing and reasoning with various
semantics of abstract argumentation. We express a wide range of
semantics within the proposed framework, including the standard
semantics due to Dung, semi-stable, stage, cf2, and resolution-based
semantics. We provide building blocks which make it easy and
straightforward to express further semantics and reasoning tasks.
Our results show that MSO can serve as a \textit{lingua franca} for abstract
argumentation that directly yields to complexity results. In particular, we obtain that for argumcntation frameworks with certain
structural properties the main computational problems with respect
to MSO-expressible semantics can all be solved in linear time. Furthermore,
we provide a novel characterisation of resolution-based grounded semantics.\\
End quote.
\end{quote}
The monadic second order theory has a symbol $R$ for the attack relation $xRy$ and quantifiers over subsets of $S\in \CA$. For all its strength, it cannot express attacks on attacks. For networks with arguments of the form $\phi(R)$ the language cannot express, e.g. $xR\phi(R)$ (i.e. that a point $x$ attacks the fact that the attack relation $R$ has the property described by $\phi$). One could extend the language of~\cite{541-22} to allow substitution of formulas like $\phi(R)$ for variables such as $y$ in $xRy$. But doing this requires an entirely new theory of self-fibring (the general complexities of which are studied in~\cite{541-23}). The solution we can adopt on the basis of the present paper is simple: we use an intuitionistic monadic second order theory. In other words, we combine out intuitionistic interpretation with that of~\cite{541-22}. This simple move also shows the value of the intuitionistic interpretation.
\end{enumerate}

\section*{Acknowledgements}

We are grateful to David Pearce for his valuable comments.

\appendix
\section*{Appendix}
\section{The logic $\BG_3$}
The logic is called $\BG_3$ because it corresponds to G\"{o}del's $n$-valued logic with $n=3$. But the logic
is due to Heyting. It was first axiomatised by Lukasiewicz in 1938 by
\[
(\neg x \to y) \to  (((y \to  x)\to  y) \to  y)
\]
But there is shorter axiom due to Umezawa later proved complete by Hosoi:
\[
x\vee (\neg y \vee (x\to y))
\]
David Pearce~\cite{541-3},  calls the logic $HT$ or `here-and-there' because of the two linearly ordered worlds.

$\BG_3$ can also be characterised by adding the following additional  two axioms to the axioms of propositional intuitionistic logic.

\paragraph{Dummett LC axiom.}
\[
(x\to y)\vee (y\to x)
\]
This gives linearity.

\paragraph{Peirce's axiom for height 2.}
\[
((x\to (((y\to z)\to y)\to y))\to x)\to x.
\]

\section{Semantics for $\BG_3$}\label{sem-prop}

We have two possible worlds $t$ and $s$, with the ordering $t<s$. $t$ is considered the actual world. The language of $\BG_3$ contains atoms and the intuitionistic connectives $\{\land,\lor,\to,\neg,\bot,\top\}$.

An assignment $h$ to the atoms gives them values in $t$ and in $s$. We write $h(x,q)\in\{\top,\bot\}$ for $x\in\{t,s\}$ and $q$ atomic. We also write $h(q)=(v_1,v_2)$ where $v_1=h(t,q)$ and $v_2=h(s,q)$. We require that
\[
h(t,q)=\top \text{ implies } h(s,q)=\top,
\]
so the value $h(q)=(\top,\bot)$ is forbidden.

We evaluate wff of the language as follows. For a given $h$, and a given $a\in\{t,s\}$
\begin{itemize}
\item
$a\vDash_h\top$
\item
$a\nvDash_h\bot$
\item
$a\vDash_h q$ iff $h(a,q)=\top$ (for any atomic $q$)
\item
$a\vDash_h A\land B$ iff  $a\vDash_h A $ and $a\vDash_h B$
\item
$a\vDash_h A\lor B$ iff  $a\vDash_h A $ or $a\vDash_h B$
\item
$a\vDash_h \neg A$ iff  $b\nvDash_h A$ for all $b\in\{t,s\}$ s.t. $a\leq b$.
\item
$a\vDash_h A\to B$ iff  $b\nvDash_h A $ or $b\vDash_h B$, for all $b\in\{t,s\}$ s.t. $a\leq b$.
\end{itemize}
We have the following theorems:
\begin{enumerate}
\item
For any $A$, if $t\vDash_h A$ then $s\vDash_h A$.
\item
$\BG_3\vdash A$ iff for any $h$, $t\vDash_h A$.
\end{enumerate}

\section{Semantics for predicate $\BG_3$}\label{sem-quan}

We consider a special version of predicate $\BG_3$, where the language has the additional quantifiers $\forall x$, $\exists x$ and the predicates $xRy$ and $\In(x)$. The logic is a constant domain logic for the worlds $\{t,s\}$ and the relation $R$ is {\em decided}, meaning $xRy$ holds at both $t$ and $s$ or neither, for any $x$ and $y$.

We thus extend the semantics of Section~\ref{sem-prop} with a domain $D$, relation $\rho\subseteq D\times D$, and a relation $\iota\subset \{t,s\}\times D$ where $\iota(t,d)\subseteq \iota(s,d)$ for $d\in D$. Assignments, which we now denote by $v$, are functions from variable symbols to elements of $D$.

We evaluate wff of the language as follows. For a given $v$, and a given $a\in\{t,s\}$
\begin{itemize}
\item
$a\vDash_v\top$
\item
$a\nvDash_v\bot$
\item
$a\vDash_v xRy$ iff $v(x)\rho v(y)$
\item
$a\vDash_v \In(x)$ iff $\iota(a,v(x))$
\item
$a\vDash_v A\land B$ iff  $a\vDash_v A $ and $a\vDash_v B$
\item
$a\vDash_v A\lor B$ iff  $a\vDash_v A $ or $a\vDash_v B$
\item
$a\vDash_v \neg A$ iff  $b\nvDash_v A$ for all $b\in\{t,s\}$ s.t. $a\leq b$.
\item
$a\vDash_v A\to B$ iff  $b\nvDash_v A $ or $b\vDash_v B$, for all $b\in\{t,s\}$ s.t. $a\leq b$.
\item
$a\vDash_v \forall x A$ iff  $b\vDash_{v[x\mapsto d]} A $ for all $d\in D$ and for all $b\in\{t,s\}$ s.t. $a\leq b$.
\item
$a\vDash_v \exists x A$ iff  $a\vDash_{v[x\mapsto d]} A $ for some $d\in D$.\footnote{Where $v[x\mapsto d]$ is a function identical to $v$ except $v(x)=d$.}
\end{itemize}

\end{document}